\documentclass[11pt]{article}
\usepackage{graphicx}
\usepackage{pgfplots}

\usepackage[colorlinks=true,citecolor=blue]{hyperref}
\usepackage{natbib,extarrows}
\usepackage{graphicx}
\usepackage{amsfonts}
\usepackage{amsmath}
\usepackage{amssymb}
\usepackage{url}
\usepackage{fancyhdr}
\usepackage{indentfirst}
\usepackage{enumerate}
\usepackage{titlesec}
\usepackage{amsthm}
\usepackage{subcaption}
\usepackage{dsfont}
\usepackage[misc]{ifsym}

\theoremstyle{definition}
\newtheorem{example}{Example}

\newtheorem{definition}{Definition}
\newtheorem{task}{Task}
\newtheorem{conjecture}[task]{Conjecture}

\theoremstyle{plain}
\newtheorem{theorem}{Theorem}
\newtheorem{lemma}{Lemma}
\newtheorem{proposition}{Proposition}
\newtheorem{corollary}{Corollary}
\theoremstyle{remark}
\newtheorem{remark}{Remark}
    \allowdisplaybreaks

\usepackage{cases}

\theoremstyle{definition}
         
\def\oo{\infty}                   \def\d{\,\mathrm{d}}
\def\lm{\lambda}                  \def\th{\theta}

\def\N{\mathbb{N}}
\def\P{\mathbb{P}}
\def\p{\mathbb{P}}
\def\E{\mathbb{E}}

\def\R{\mathbb{R}}

\def\X{\mathcal{X}}

\def\var{\mathrm{var}}

\def\X{\mathcal{X}}

\def\d{\mathrm{d}}

\def\oo{\infty}

\renewcommand{\(}{\left(}
\renewcommand{\)}{\right)}

\usepackage[onehalfspacing]{setspace}

\usepackage{bm}
\usepackage{tikz-qtree}
\usepackage{tikz}

\def\id{\mathds{1}}

\begin{document}

\title{Diversification for infinite-mean Pareto models without risk aversion}


\author{Yuyu Chen\thanks{Department of Economics, University of Melbourne,  Australia. \Letter~{\scriptsize\url{yuyu.chen@unimelb.edu.au}}}
\and Taizhong Hu\thanks{Department of Statistics and Finance, University of Science and Technology of China,  China. \Letter~{\scriptsize\url{thu@ustc.edu.cn}}}
\and Ruodu Wang\thanks{Department of Statistics and Actuarial Science, University of Waterloo,  Canada. \Letter~{\scriptsize\url{wang@uwaterloo.ca}}}
\and Zhenfeng Zou\thanks{Department of Statistics and Finance, University of Science and Technology of China,  China. \Letter~{\scriptsize\url{zfzou@ustc.edu.cn}}}
}
\date{\today}
\maketitle

\begin{abstract}

We study stochastic dominance between portfolios of independent and identically distributed (iid) extremely heavy-tailed (i.e., infinite-mean) Pareto random variables.  
With the notion of majorization order, we show that a more diversified portfolio of iid extremely heavy-tailed Pareto random variables is larger in the sense of first-order stochastic dominance. This result is further generalized for Pareto random variables caused by triggering events, random variables with tails being Pareto, bounded Pareto random variables, and positively dependent Pareto random variables. These results provide an important implication in investment: Diversification of extremely heavy-tailed Pareto profits uniformly increases investors' profitability, leading to a diversification benefit. 
Remarkably, different from the finite-mean setting, such a diversification benefit does not depend on the decision maker's risk aversion.

\end{abstract}

\section{Introduction}

Pareto distributions are ubiquitous for modeling heavy-tailed phenomena in diverse fields;  \cite{andriani2007beyond} presented over 80 applications of Pareto distributions. Their applications in economics, finance, and insurance can be found in, e.g., \cite{A15}, \cite{MFE15} and \cite{ibragimov2015heavy}. 
Among these applications, some datasets modelled by Pareto distributions exhibit extremely heavy tails, leading to infinite first moments. Such examples include financial returns from some technological innovations \citep{scherer2001uncertainty,silverberg2007size}, corporate innovative outputs \citep{choi2020power},  catastrophic losses  \citep{IJW09,hofert2012statistical}, and operational losses \citep{moscadelli2004modelling}.   A list of empirical examples of infinite-mean Pareto models is summarized by \cite{CEW24(2)}.

Diversification is a fundamental concept in finance and economics, pioneered by \cite{M52}. 
At a high level, when it comes to decision making in portfolio selection,  it is commonly agreed  
that risk-averse agents (in the sense of \cite{RS70}) would prefer more diversification, whereas risk-seeking agents
do not like diversification; e.g., \cite{S67} and \cite{FM72}. Such results are typically formulated for finite-mean models of investment payoffs, but for infinite-mean models, notions of risk aversion and risk seeking, which rely on expected utility functions, may not behave well and can be difficult to interpret.
Intriguing or even counter-intuitive results may be obtained for infinite-mean models, as recently studied by \cite{CEW24}.

This paper focuses on the portfolio diversification of extremely heavy-tailed, i.e., infinite-mean, Pareto random variables.
For two random variables $X$ and $Y$, we say that $X$ is smaller than $Y$ in \emph{first-order stochastic dominance}, denoted by $X\le_{\rm st}Y$, if $\p\(X\le x\)\ge \p\(Y\le x\)$ for all $x\in\R$. First-order stochastic
dominance is the strongest commonly used stochastic dominance ordering. For iid extremely heavy-tailed Pareto random variables $X_{1},\dots,X_{n}$ and a nonnegative vector $\(\theta_{1},\dots,\theta_n\)  $, \cite{CEW24} showed  
\begin{equation}
 \label{eq:*}
      \theta X_{1}\le_{\rm st}\theta_{1}X_{1}+\dots+\theta_{n}X_{n} \mbox{~~~~where $\theta =\sum_{i=1}^n\theta_n $.}
\end{equation}
Intuitively,  \eqref{eq:*} suggests that a diversified portfolio of iid extremely heavy-tailed Pareto random variables has a larger tail probability than a non-diversified portfolio.
If $X_1,\dots,X_n$ represent profits, then any  decision maker consistent with $\le_{\rm st}$ will choose the diversified portfolio represented by $(\theta_1,\dots,\theta_n)$ over the non-diversified one represented by $(\theta,0,\dots,0)$.
The situation will be completely flipped if $X_1,\dots,X_n$ represent losses; the non-diversified portfolio is now preferred.
Therefore, whether agents are risk-averse or not does not matter for the decision, but whether the infinite-mean random variables represent gains or losses matters.

There is a clear limitation of \eqref{eq:*}: It only allows for the comparison of a general portfolio with a completely non-diversified portfolio. To understand decisions in portfolio selection, 
one needs to compare a portfolio with another one, possibly also diversified, but to a lesser extent. This problem was not addressed by the results of \cite{CEW24}. Thus, we answer the following question in this paper: 
What is the stochastic dominance relation between two different diversified portfolios of iid extremely heavy-tailed Pareto random variables? 

Let $\(\theta_{1},\dots,\theta_n\)$ and $\(\eta_1, \dots, \eta_n\)\in\R_+^n$ be the exposure vectors of the two portfolios  such that $\(\theta_1, \dots, \theta_n\)$ is smaller than $\(\eta_{1}, \dots, \eta_n\)$ in majorization order (see Section \ref{sec:2} for the definition of majorization order). This implies that the two vectors have equal total exposures (i.e., $\sum_{i=1}^n\theta_i=\sum_{i=1}^n\eta_i$), making it fair to compare the two portfolios. Moreover, the components of a smaller vector in majorization order are less ``spread out''; that is, the components of $\(\theta_1, \dots,\theta_n\)$ are closer to each other than those of $\(\eta_1, \dots,\eta_n\)$.  Therefore,  it is natural to consider the portfolio with exposure vector $\(\theta_1, \dots,\theta_n\)$ as more diversified than that with exposure vector $\(\eta_1, \dots,\eta_n\)$. One of our main results, presented in Section \ref{sec:2}, is the following inequality 
\begin{align}
 \label{eq:main}
      \eta_1 X_1+\dots+\eta_n X_n \le_{\rm st} \theta_1 X_1 +\dots+\theta_n X_n.
\end{align}
Inequality \eqref{eq:main}, in the sense of the strongest form of risk comparison,  shows that a more diversified portfolio of iid extremely heavy-tailed Pareto random variables leads to a larger tail probability. Note that \eqref{eq:main} includes \eqref{eq:*} as a special case.

In Section \ref{sec:3}, we generalize \eqref{eq:main} into several models that can be more relevant in different settings. In particular, in Theorem \ref{th:cata}, extremely heavy-tailed Pareto random variables are assumed to be caused by triggering events, which usually happen with small probabilities. Proposition \ref{prop:tail} considers iid random variables whose tails follow an extremely heavy-tailed Pareto distribution. Proposition \ref{prop:bounded} studies extremely heavy-tailed Pareto distributions with an upper bound. Theorem \ref{thm:Clayton} deals with extremely heavy-tailed Pareto random variables that are positively dependent and modelled by a Clayton copula.

The implications of \eqref{eq:main} in investment decisions are discussed in Section \ref{sec:4}. As expected, for agents making investments  whose profits are iid extremely heavy-tailed Pareto random variables,  diversification uniformly improves their preferences as long as the preferences are monotone and well-defined. In particular, they do not need to be risk averse. 
As such improvements are strict, the optimal investment strategy is to assign an equal weight to each investment in the portfolio. This strategy, known as the $1/n$ rule, is also optimal for risk-averse expected utility agents if the profits have finite mean.
On the other hand, if the random variables represent losses, then the optimal decision is the opposite, as implied by \eqref{eq:*}.
Section \ref{sec:5} concludes the paper.

There are some studies on the diversification of extremely heavy-tailed random variables (see, e.g., \cite{embrechts2002correlation} and \cite{ibragimov2009portfolio}), and most of their findings concern how diversification affects the spreadness of portfolios and thus are different from \eqref{eq:main}.  
Although \eqref{eq:*}  is a special case of \eqref{eq:main}, our technical results are  quite different from those of \cite{CEW24}. First, the results in the latter paper do not address the comparison between diversified portfolios with different weight vectors, and thus they are weaker than ours when formulated with the same dependence assumption.
Our proofs are distinct
from \cite{CEW24} and rely on analysis of majorization order and probabilistic inequalities.
On the other hand, \cite{CEW24} showed \eqref{eq:*} under some form of negative dependence. It remains unclear whether our main results can be generalized to accommodate negative dependence, due to distinct mathematical techniques used in the proofs.
Instead, we establish a result under a specific form of positive dependence in Section \ref{sec:3.4}. Another related result is obtained by \cite{ibragimov2005new}, which implies that \eqref{eq:main} holds for iid positive one-sided stable random variables with infinite mean.  The techniques of \cite{ibragimov2005new} rely on properties of stable distributions, which are very different from ours, as Pareto distributions do not belong to that class.

We conclude this section by fixing some notation. Throughout, random variables are defined on an atomless probability space $\(\Omega,\mathcal F,\p\)$. Denote by $\mathbb N$ the set of all positive integers and $\R_+$ the set of non-negative real numbers. For $n\in\N$, write $[n]=\{1,\dots,n\}$. For vectors $\bm x=\(x_1,\dots,x_n\)\in\R^n$ and $\bm y=\(y_1,\dots,y_n\)\in\R^n$, their dot product is $\bm x\cdot \bm y=\sum_{i=1}^nx_iy_i$. For $\bm x= \(x_1,\dots,x_n\)\in\R^n$, denote by $\|\bm x\|=\sum_{i=1}^n|x_i|$. For $x,y\in \R$, write $x\wedge y= \min\{ x,y\}$, $x\vee y= \max\{ x,y\}$, and $x_+=\max\{x,0\}$. We write $\bm 1_n=(1,\ldots, 1)\in\R^n$ for $n\in\N$. For a random variable $X$ with distribution $F$, its quantile is defined as
\begin{equation*}
 \label{VaR}
 F^{-1}(p) :=\inf\left\{t\in\R:F(t)\ge  p\right\},~~~ p\in(0,1).
\end{equation*}
In this paper, terms like ``increasing" and ``decreasing" are in the non-strict sense.

\section{Diversification of iid  Pareto random variables}
\label{sec:2}

\subsection{Majorization order}

For $\theta,\alpha>0$, the Pareto distribution is  specified by the cumulative distribution function 
\begin{align*}
  P_{\alpha,\theta}(x) = 1 -\left(\frac{\theta}{x}\right)^{\alpha},\quad x\ge \theta.
\end{align*}
The parameter $\theta$ is a scale parameter (i.e.,  for $X\sim P_{\alpha,1}$, we have $\theta X\sim P_{\alpha,\theta}$) and $\alpha$ is the tail parameter. The mean of $P_{\alpha,\theta}$ is infinite if and only if $\alpha\in (0,1]$. Clearly, it suffices to consider the $P_{\alpha,1}$ distribution which we write as $\mathrm {Pareto}(\alpha)$.  We say that the $\mathrm {Pareto}(\alpha)$ distribution is \emph{extremely heavy-tailed} if $\alpha\le 1$, and it is \emph{moderately heavy-tailed} if $\alpha >1$. A profit or loss that follows a $\mathrm {Pareto}(\alpha)$ distribution will be referred to as \emph{Pareto profit} or \emph{Pareto loss}, respectively.

We recall the notions of $T$-transform and majorization order from \cite{MOA11}. 
\begin{definition}
\label{def:equitable}
Let $\bm {\theta}=\(\theta_1,\dots,\theta_n\)$ and $\bm {\eta} = \(\eta_1,\dots,\eta_n\)$ be two vectors in $\R^n$.
\begin{enumerate}[(i)]
  \item $\bm {\theta}$ is a \emph{$T$-transform} of $\bm {\eta}$ if for some $\lambda \in [0,1]$ and $i,j\in [n]$ with $i \neq j$,
	\begin{equation}\label{eq:t-tran}
	    \theta_i=\lambda \eta_i + (1-\lambda)\eta_j,~\theta_j= (1-\lambda)\eta_i+\lambda \eta_j,
	\end{equation}
	and $\theta_k=\eta_k$ for $k\in [n]\setminus \{i,j\}$.

  \item  $\bm \theta$ is dominated by $\bm \eta$ in \emph{majorization order}, denoted by $\bm \theta \preceq \bm \eta$, if $\sum^n_{i=1}\theta_i =\sum^n_{i=1}\eta_i$ and
      \begin{equation*}
         \sum^k_{i=1} \theta_{(i)} \ge \sum^k_{i=1} \eta_{(i)}\ \ \mbox{for}\ k\in [n-1],
      \end{equation*}
 
      where $\th_{(i)}$ and $\eta_{(i)}$ represent the $i$th smallest order statistics of $\bm \th$ and $\bm \eta$, respectively. We write $\bm \theta \prec \bm\eta$ if $\bm\theta \preceq \bm\eta$ and $\bm\theta \neq \bm\eta$.
\end{enumerate}
\end{definition}
The lemma below will be used to establish our main result.
\begin{lemma}
 \label{lem:maj}
{\rm \citep[][Section 1.A.3]{MOA11}}\ \
For $\bm\theta, \bm \eta\in \R^n$, $\bm \th\preceq \bm \eta$ if and only if $\bm \theta$ can be obtained from $\bm\eta$ by a finite number of $T$-transforms.
\end{lemma}

Clearly, for $\bm\th , \bm\eta\in\R_+^n$ with $\bm\th\preceq \bm\eta$, the components of $\bm\theta$ are less ``spread out'' than those of $\bm\eta$. For instance,
$$
   \(\frac{\|\bm\th\|}{n}, \dots, \frac{\|\bm\th\|}{n}\) \preceq \bm\th \preceq (\|\bm\th\|, 0, \dots, 0) \ \ \mbox{for all}\ \bm\th\in\R_+^n.
$$
Therefore, if $\bm\theta$ and $\bm\eta$ are the exposure vectors of two portfolios, the portfolio with exposure vector $\bm\th$ is considered to be more diversified than the one with exposure vector $\bm\eta$. 

\subsection{Main result}

For $\bm a=\(a_1,\dots, a_n\)\in\R^n$, denote its increasing rearrangement by $\bm {a}^\uparrow = \left(a_{(1)}, a_{(2)}, \dots, a_{(n)}\right)$. For random variables $X$ and $Y$, we write $X<_{\rm st}Y$ if $\p(X>x)<\p(Y>x)$ for all $x\in \R$ satisfying $\p(X>x)\in (0,1)$, and this will be referred to as strict stochastic dominance.\footnote{This condition is stronger than a different notion of strict stochastic dominance defined by $X\le _{\rm st} Y$ and $X\not \ge_{\rm st} Y$.} In the following lemma, we first compare portfolios of two independent extremely heavy-tailed Pareto random variables with possibly different tail parameters.

\begin{lemma}
 \label{le-3-1}
Let $\bm X =\(X_1, X_2\)$ be a vector of independent components with $X_i \sim {\rm Pareto}(\alpha_i)$ for  $i=1,2$ where $0<\alpha_1 \le  \alpha_2\le 1$. For $\bm\th, \bm\eta\in \R_+^2$ satisfying $\bm\th \preceq \bm\eta$, we have $\bm{\th}^\uparrow \cdot \bm X  \ge_{\rm st}  \bm{\eta}^\uparrow \cdot \bm X  $. Moreover, if $\bm\th \prec \bm\eta$, then $\bm{\th}^\uparrow \cdot \bm X >_{\rm st} \bm{\eta}^\uparrow \cdot \bm X $. 
\end{lemma}

\begin{proof}
Without loss of generality, assume that $||\bm\theta||=||\bm\eta||=1$ and  $\bm\theta$ is a $T$-transform of $\bm\eta$. Write $\bm\theta^\uparrow= \(\theta,1-\theta\)$ and $\bm\eta^\uparrow=\(\eta,1-\eta\)$.
As $\bm\theta$ is a $T$-transform of $\bm\eta$, we have $0\le \eta\le \th\le 1/2$. Let $r_i = -1/\alpha_i \in (-\oo, -1]$ for $i = 1, 2$. Then, there exist $U_1, U_2 \overset{{\rm iid}}{\sim} \mathrm{U}(0, 1)$ such that $X_i \overset{\d}{=} U_i^{r_i}$ for $i=1, 2$. For the first statement on stochastic dominance, it suffices to show that, for $0<\eta <\th\le 1/2$,
\begin{equation}
 \label{eq-230305}
   \eta U_1^{r_1}+(1-\eta) U_2^{r_2}\le_{\rm st}\th U_1^{r_1}+(1-\th) U_2^{r_2}.
\end{equation}
We first derive the distribution function of $\eta U_1^{r_1}+ (1-\eta) U_2^{r_2}$. For any $x \ge 1$,
\begin{align*}
	\P\(\eta U_1^{r_1} + (1-\eta) U_2^{r_2} \le x\)
     &= \int_{\left(\frac{x-\eta}{1-\eta}\right)^{1/r_2}}^1 \int_{\left(\frac{x-(1-\eta)y^{r_2}}{\eta} \right)^{1/r_1}}^1 \d u_1 \d y\\
 	&= 1 - \left(\frac{x-\eta}{1-\eta}\right)^{1/r_2} - \int_{\left(\frac{x-\eta}{1-\eta}\right)^{1/r_2}}^1
           \left(\frac{x-(1-\eta)y^{r_2}}{\eta}\right)^{1/r_1} \d y.
\end{align*}
Similarly, we have
$$
	\P\(\th U_1^{r_1} + (1-\th) U_2^{r_2} \le x\) = 1 - \left(\frac{x-\th}{1-\th}\right)^{1/r_2} - \int_{\left(\frac{x-\th}{1-\th}\right)^{1/r_2}}^1 \left(\frac{x-(1-\th)y^{r_2}}{\th}\right)^{1/r_1} \d y.
$$
Let $x\ge 1$. For $z\in (0,1/2]$,
define $$
	H(z)= \left(\frac{x-z}{1-z}\right)^{1/r_2} +
        \int_{\left(\frac{x-z}{1-z}\right)^{1/r_2}}^1
            \left(\frac{x-(1-z)y^{r_2}}{z}\right)^{1/r_1} \d y.
$$
To prove \eqref{eq-230305}, we need to show   $H(\eta)\le H(\theta)$ for any $x\ge 1$. It suffices to show that $H(z)$ is increasing in $z\in(0,1/2]$. The derivative of $H(z)$ is
\begin{align*}
   H'(z) &= \frac{1}{r_1}\int_{\left(\frac{x-z}{1-z}\right)^{1/r_2}}^1 \left(\frac{x-(1-z)y^{r_2}}{z}\right)^{1/r_1-1} \frac{y^{r_2}-x}{z^2} \d y \\
   & = \frac{1}{r_1r_2z^{1/r_1+1}} \int_{\frac{x-z}{1-z}}^1 \left(x-(1-z)t\right)^{1/r_1-1} (t-x) t^{1/r_2-1} \d t \\
   & = \frac{1}{r_1r_2z^{1/r_1+1}(1-z)^{1/r_2+1}} \int_{x-z}^{1-z} \left(x-v\right)^{1/r_1-1} (v-(1-z)x) v^{1/r_2-1} \d v \\
	&= \frac{x^{1/r_1+1/r_2}}{r_1r_2z^{1/r_1+1}(1-z)^{1/r_2+1}} \int^{(x-z)/x}_{(1-z)/x} \left(1-y\right)^{1/r_1-1} (1-z-y) y^{1/r_2-1} \d y.
\end{align*}
Thus, it suffices to prove  
\begin{equation*}\label{emph}
g(x):=\int^{(x-z)/x}_{(1-z)/x} \left(1-y\right)^{1/r_1-1} y^{1/r_2-1}(y-1+z) \d y \le 0.
\end{equation*}
Since $g(1)=0$, we need to show that $g'(x)\le 0$ for $x>1$. Note that $1/r_2 \le 1/r_1$ by $\alpha_1 \le \alpha_2$. Thus,
\begin{align*}
    g'(x) &= \left(1-\frac{x-z}{x}\right)^{1/r_1-1} \left(\frac{x-z}{x}\right)^{1/r_2-1}
         \left(\frac{x-z}{x}-1+z\right)\frac{z}{x^2} \\
    & \quad\quad+ \left(1-\frac{1-z}{x}\right)^{1/r_1-1} \left(\frac{1-z}{x}\right)^{1/r_2-1} \left(\frac{1-z}{x}-1+z\right)\frac{1-z}{x^2} \\
	&=  \left[z^{1/r_1+1}(x-z)^{1/r_2-1}-(1-z)^{1/r_2+1}(x-1+z)^{1/r_1-1}\right]
           \frac{x-1}{x^{1/r_1 + 1/r_2 +1}} \\
	&\le \left[z^{1/r_1-1/r_2}(x-z)^{1/r_2-1}-(x-1+z)^{1/r_1-1}\right](1-z)^{1/r_2+1}
            \frac{x-1}{x^{1/r_1+1/r_2+1}} \\
	&\le \left[z^{1/r_1-1/r_2}-(x-1+z)^{1/r_1-1/r_2}\right](x-1+z)^{1/r_2-1}
          (1-z)^{1/r_2+1}\frac{x-1}{x^{1/r_1+1/r_2+1}}  \le 0.
\end{align*}
This means that $g(x)$ is decreasing, and hence $g(x) \le g(1)=0$ for $x \ge 1$. It is clear that $H(z$) is strictly increasing in $z\in(0,1/2)$ for $x>1$. Hence, the strict inequality in the second assertion of the lemma is obtained.
\end{proof}

Now, using Lemmas \ref{lem:maj} and \ref{le-3-1} and the fact that first-order stochastic dominance is closed under convolution (e.g., Theorem 1.A.3 (b) of \cite{SS07}), we obtain the following theorem for $n>2$ independent extremely heavy-tailed Pareto random variables.

\begin{theorem}\label{th:casen}
Let $\bm X =\(X_1, \dots, X_n\)$ be a vector of independent components with $X_i \sim {\rm Pareto}(\alpha_i)$ for all $i\in[n]$ where $0<\alpha_1 \le \dots \le \alpha_n\le 1$. For $\bm\th, \bm\eta\in \R_+^n$ satisfying $\bm\th \preceq \bm\eta$, we have $\bm{\th}^\uparrow \cdot \bm X  \ge_{\rm st}  \bm{\eta}^\uparrow \cdot \bm X  $. Moreover, if $\bm\th \prec \bm\eta$, then $\bm{\th}^\uparrow \cdot \bm X >_{\rm st} \bm{\eta}^\uparrow \cdot \bm X $. In particular, if $\alpha_1 = \dots = \alpha_n$, we have $\bm{\theta} \cdot \bm X \ge_{{\rm st}}\bm{\eta} \cdot\bm X$ when $\bm\th \preceq \bm\eta$ and $\bm{\th} \cdot \bm X >_{\rm st} \bm{\eta} \cdot \bm X$ when $\bm\th \prec \bm\eta$. 
\end{theorem}

\begin{proof}
By Lemma 1 in \cite{Ma98}, there exists a finite number of vectors $\bm{c}_j^\uparrow$, which is an increasing rearrangement of $\bm c_j \in \R^n_+$, $j\in[N]$, such that
$$
     \bm{\eta}^\uparrow = \bm{c}_1^\uparrow \succeq \bm{c}_2^\uparrow \succeq \dots \succeq \bm{c}_N^\uparrow = \bm{\theta}^\uparrow,
$$
where $\bm{c}_{j}^\uparrow$ is a $T$-transform of $\bm{c}_{j-1}^\uparrow$ for each $j\in[N]\backslash \{1\}$. Hence, it suffices to show that $\bm\th^\uparrow \cdot \bm X \ge_{\rm st} \bm{\eta}^\uparrow \cdot\bm X$ holds for $\bm\th^\uparrow$ being a $T$-transform of $\boldsymbol{\eta}^\uparrow$. Suppose that $\bm\th =(\th_1, \dots, \th_n)$ and $\bm\eta=(\eta_1,\dots,\eta_n)$  such that $\bm\th \preceq \bm\eta$. Take $i<j $ as in \eqref{eq:t-tran} such that $(\theta_{(i)},\theta_{(j)})$ is a $T$-transform of $(\eta_{(i)},\eta_{(j)})$. By  Lemma \ref{le-3-1}, we have $\theta_{(i)} X_i +\theta_{(j)} X_j \ge_{\rm st} \eta_{(i)} X_i+\eta_{(j)}X_j$. As first-order stochastic dominance is closed under convolutions and $\theta_{(k)}=\eta_{(k)}$ for $k\in [n]\setminus\{i,j\}$, we obtain 
$$\sum_{k=1}^n \theta_{(k)}X_k \ge_{\rm st}\sum_{k=1}^n \eta_{(k)}X_k.$$ 
Moreover, let $c=\theta_{(i)}+\theta_{(j)}$ and
$$
    S_{i,j}:=\sum_{k\in[n]\setminus\{i,j\}} \theta_{(k)}X_k=\sum_{k\in[n]\setminus\{i,j\}} \eta_{(k)}X_k,
$$
for which the density function is denoted by $f$.  For $t>\|\bm\th\|$, we have
\begin{align*}
  &\p\(\bm\th^\uparrow \cdot \bm X>t\)
	\\ & =\p\(S_{i,j}+\theta_{(i)}X_i+\theta_{(j)}X_j>t\)\\
	& =\int_{\|\bm\th\|-c}^\infty\p\(\theta_{(i)}X_i+\theta_{(j)}X_j>t-s\)f(s)\,\d s\\
    & =\int_{\|\bm\th\|-c}^{t-c}\p\(\theta_{(i)}X_i+\theta_{(j)}X_j>t-s\)f(s)\, \d s  +\int_{t-c}^\infty\p\(\theta_{(i)}X_i+\theta_{(j)}X_j>t-s\)f(s)\,\d s\\
    & >\int_{\|\bm\eta\|-c}^{t-c}\p\(\eta_{(i)}X_i+\eta_{(j)}X_j>t-s\)f(s)\,\d s    +\int_{t-c}^\infty\p\(\eta_{(i)}X_i+\eta_{(j)}X_j>t-s\)f(s)\,\d s\\
	&=\p(\bm\eta^\uparrow \cdot \bm X>t),
\end{align*}
where the inequality follows from the strictness statement in Lemma \ref{le-3-1}, i.e., $\bm{\th}^\uparrow \cdot \bm X >_{\rm st} \bm{\eta}^\uparrow \cdot \bm X$.
\end{proof}

Note that Pareto distributions with smaller tail indices are larger in first-order stochastic dominance. In the case that the Pareto random variables have non-identical tail indices, Theorem \ref{th:casen} means that if we shift exposures from extremely heavy-tailed Pareto random variables with larger tail indices to those with smaller tail indices, the portfolio gets larger in first-order stochastic dominance.

   Theorem \ref{th:casen} is a natural generalization of the main results of \cite{CEW24}, and there are several  distinctions. First,   Theorem  \ref{th:casen} allows for the comparison of two diversified portfolios.
 Second, 
 Theorem \ref{th:casen} considers independent but not necessarily identically distributed Pareto random variables. 
  Third, completely different proof techniques are used for  Theorem \ref{th:casen} 
  and the other paper. 

For the rest of the paper, we will mainly focus on the case of identically distributed Pareto random variables, to get concise results and to separate the effect of diversification from the effect of marginal distributions.

Theorem \ref{th:casen} shows that a portfolio of iid extremely heavy-tailed Pareto random variables with a smaller exposure vector in majorization order is larger in first-order stochastic dominance. This implies a  \emph{diversification benefit}: A more diversified portfolio of iid extremely heavy-tailed Pareto profits is strictly preferred; the strictness of the diversification benefit is crucial for the optimal investment decisions of agents (see Section \ref{sec:4}). Several generalizations of Theorem \ref{th:casen} are studied in Section \ref{sec:3}.

\subsection{Further discussions}

The assumption of an extremely heavy tail is necessary for Theorem \ref{th:casen} in the iid case, which is made clear by the following proposition, with its proof provided in the appendix.
\begin{proposition}
   \label{prop:infinite}
Suppose that $\bm\th, \bm\eta\in \R_+^n$ such that $\bm\th \prec \bm\eta$, and $\bm X$ is a vector of $n$ iid random variables with finite mean. The inequality $\bm\th \cdot \bm X \ge_{\rm st} \bm\eta\cdot \bm X$ cannot hold unless $\mathbf X $ is degenerate. \end{proposition}
Proposition \ref{prop:infinite} suggests that first-order stochastic dominance is not able to compare portfolios of iid finite-mean random variables. A more suitable notion of stochastic dominance in this case is second-order stochastic dominance. For two random variables $X$ and $Y$, $X$ is smaller than $Y$ in \emph{second-order stochastic dominance}, denoted by $X\le_{\rm ssd}Y$, if $\E(u(X))\le \E(u(Y))$ for all increasing and concave functions $u$, provided that the expectations exist. The relation $X\le_{\rm ssd}Y$ means that $Y$ is ``less variable" than $X$. The following result directly follows from Theorem 3.A.35 of \cite{SS07}.
\begin{proposition}\label{prop:finite}
Let $\bm X$ be a vector of $n$ iid Pareto random variables with tail parameter $\alpha>1$. For $\bm\th, \bm\eta\in \R_+^n$ satisfying $\bm\th \preceq \bm\eta$, we have $\bm \eta\cdot \bm X\le_{\rm ssd}\bm \theta\cdot \bm X$.
\end{proposition}
Proposition \ref{prop:finite} means that diversification of iid Pareto random variables with finite mean makes the portfolio less ``spread out", and thus less risky in the sense of \cite{RS70}. This observation is distinct from Theorem \ref{th:casen}, which suggests that diversification makes a portfolio ``larger" if the Pareto random variables are extremely heavy-tailed. Implications of Theorem \ref{th:casen} and Proposition \ref{prop:finite} for investors are discussed in Section \ref{sec:4}.

\begin{remark}[Generalized Pareto distributions]\label{remark:GPD}
 As a pillar of the Extreme Value Theory, the Pickands-Balkema-de Haan Theorem \citep{BD74, P75} states that the generalized Pareto distributions are the only possible non-degenerate limiting distributions of the excess of random variables beyond a high threshold.
The generalized Pareto distribution for $\xi\ge0$ is defined as
\begin{equation*}\label{eq:GPD}
    G_{\xi,\beta}(x)=1-\(1+\xi\frac{x}{\beta}\)^{-1/\xi},~~~ x\ge 0,
\end{equation*}
where $\xi\ge 0$ ($\xi=0$ corresponds to an exponential distribution) and $\beta>0$; see \cite{EKM97}. If $\xi\ge 1$, then $G_{\xi,\beta}$ does not have finite mean. As a generalized Pareto distribution for $\xi>0$ can be converted to  $P_{1/\xi,1}$  by a location-scale transform,  Theorem \ref{th:casen} can be equivalently stated for the generalized Pareto distribution with $\xi\ge1$.
\end{remark}

\begin{remark}[Extremely heavy-tailed Pareto sum]
We say an \emph{extremely heavy-tailed Pareto sum} is a random variable $\sum_{j\in\N} \lm_j Y_j$
where $Y_j\sim \mathrm{Pareto}(\alpha_j)$, $j\in \N$, are independent, $\alpha_j\in (0,1]$, $\lm_j \in \R_+$,
and $\sum_{j\in\N} \lm_j<\infty$. Let $\bm X $ be a vector of iid extremely heavy-tailed Pareto sum random variables,
and $\bm\th, \bm\eta\in \R_+^n$ satisfy $\bm\th \preceq \bm\eta$. Then $\bm\th \cdot \bm X \ge_{\rm st} \bm\eta\cdot \bm X$. This can be shown by applying Theorem \ref{th:casen} to iid copies of each $Y_j$ and using the fact that convolution preserves first-order stochastic dominance.
\end{remark}

\begin{remark}
Let $X_1 \sim {\rm Pareto(\alpha_1)}$ and $X_2 \sim {\rm Pareto(\alpha_2)}$ be independent with $\alpha_1, \alpha_2 \in (0,1]$ such that $\alpha_1 \le \alpha_2$. For $(\eta_1, \eta_2),(\theta_1, \theta_2)\in\R^2$ such that  $(\eta_1, \eta_2) \succeq (\th_1, \th_2)$, it is natural to consider whether the following statement is true
\begin{equation} \label{eq-6}
	\eta_{(2)} X_1 + \eta_{(1)} X_2 \le_{{\rm st}} \th_{(2)} X_1 +\th_{(1)} X_2.
\end{equation}
Fix $\alpha_1 = 0.15$ and $\alpha_2=0.75$. 
In Figure \ref{fig:1},  we plot the empirical distribution functions of $\eta_{(2)} X_1 + \eta_{(1)} X_2$ and $\theta_{(2)} X_1 + \theta_{(1)} X_2$ for two cases: $(\eta_1, \eta_2, \theta_1, \theta_2) = (6, 2, 5, 3)$ and $(\eta_1, \eta_2, \theta_1, \theta_2) = (9, 1, 6, 4)$. Based on our numerical result, \eqref{eq-6} does not seem to hold, although we do not have a proof.

\begin{figure}[bht]
		\centering
		\includegraphics[width=10cm]{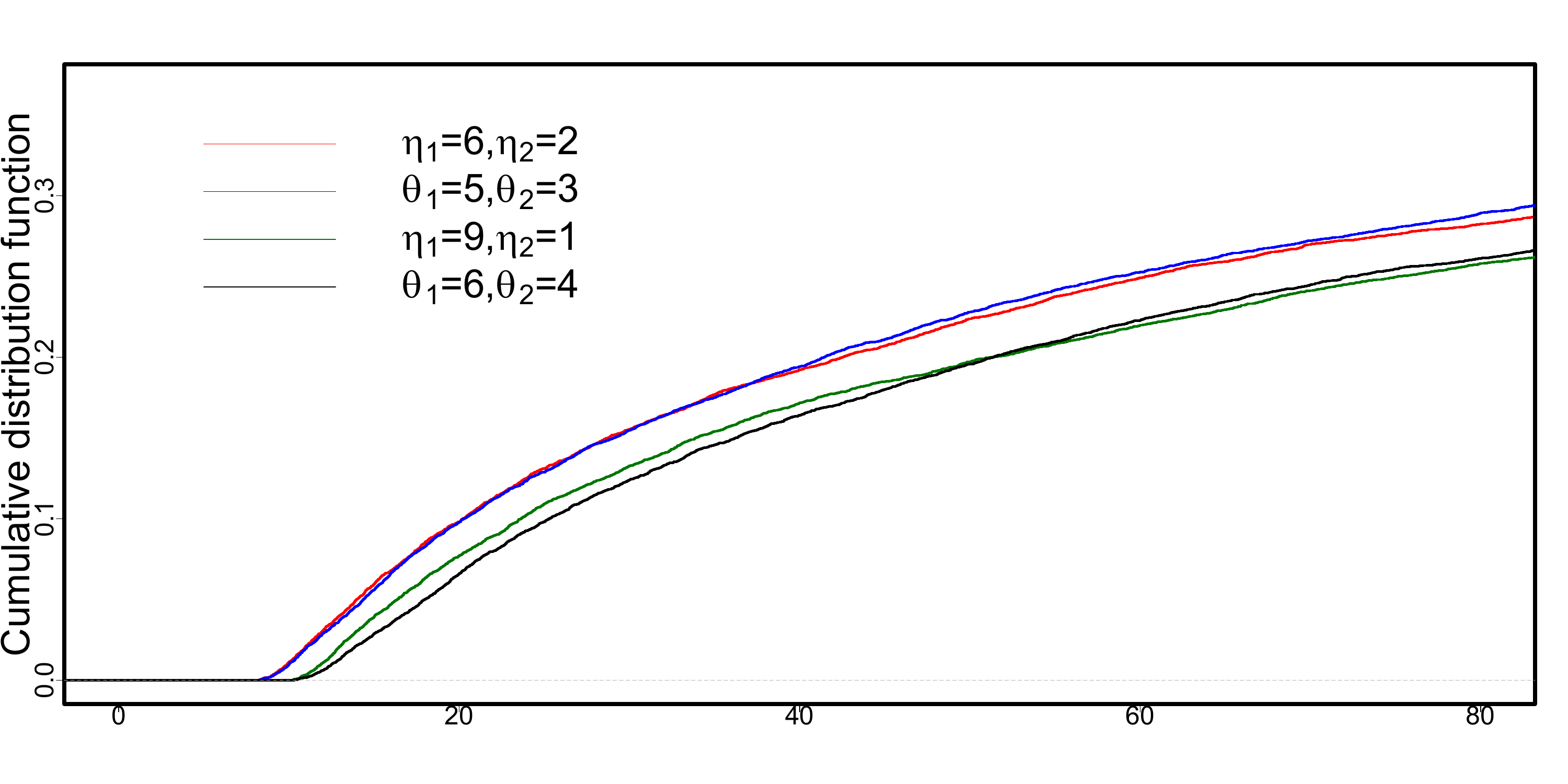}
    \caption{Empirical cumulative distribution functions of $\eta_{(2)} X_1+\eta_{(1)} X_2$ and $\th_{(2)} X_1+ \th_{(1)} X_2$ using $\alpha_1=0.15$ and $\alpha_2=0.75$}
    \label{fig:1}
\end{figure}
\end{remark}

Next, we present a simple corollary on comparing equally weighted portfolios with different sizes. Write $\bm \gamma_n=\bm 1_n/n$ for $n\in\N$. By noting that
$$
    \(\frac{1}{n+1},\dots,\frac{1}{n+1},\frac{1}{n+1}\)\preceq  \boldsymbol\(\frac{1}{n},\dots,\frac{1}{n},0\)\preceq \dots\preceq (1,\dots,0,0) \mbox{~for~}n\ge 1,
$$
the following result directly follows from Theorem \ref{th:casen}.

\begin{corollary}
 \label{cor:equal}
For $n\in\N$, let  $\bm X_n $  be a vector of $n$  iid $\mathrm{Pareto}(\alpha) $ random variables with $\alpha\in(0,1]$. For $k,\ell\in \mathbb N$ with $k< \ell $, we have $\bm \gamma_k \cdot \bm X_k <_{\rm st} \bm \gamma_\ell \cdot \bm X_\ell$.
\end{corollary}

Corollary \ref{cor:equal} means that the more components in a portfolio of iid extremely heavy-tailed Pareto profits, the more benefit we have from diversification.  For $\alpha\in(0,1]$, $p\in(0,1)$, and iid $\mathrm {Pareto}(\alpha)$ random variables $X_{1},\dots,X_{n}$,
denote by $F$ the distribution function of $\bm \gamma_n \cdot \bm X_n$.
We plot $F^{-1}$ for $n=2,\dots,6$ in Figure \ref{f1}. 
The difference between the curves for different $n$ becomes more pronounced  as $\alpha$ becomes smaller, i.e., the tail of the Pareto random variables becomes heavier.
 \begin{figure}[h]
\centering
\includegraphics[height=5cm, trim={0 0 0 20},clip]{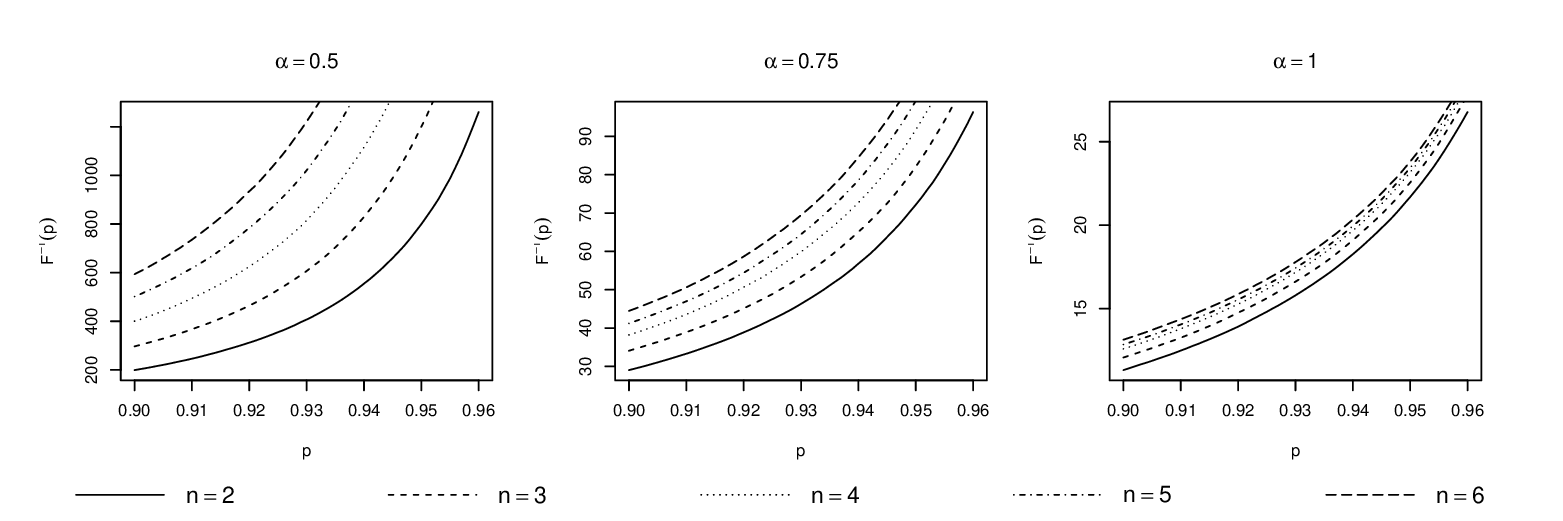}
\caption{$F^{-1}(p)$ for $n=2,\dots,6$ and $p\in(0.9,0.96)$.}
\label{f1}
\end{figure}

 \subsection{The St.~Petersburg   lottery}
The St.~Petersburg paradox is perhaps the most well-known   example of an infinite-mean profit model in economic theory.
The St.~Petersburg lottery has a distribution that pays $2^k$ with probability $2^{-k}$ for each $k\in\N$.
This distribution can be seen as the discrete version of Pareto(1), because
$
\p(X>2^k)= 2^{-k}
$ for $k\in \N$ if $X$ is the St.~Petersburg lottery, 
and 
$
\p(X>2^k)= 2^{-k}
$ for all $k\in(0,\infty)$
if $X$ follows a Pareto(1) distribution.
One may naturally wonder whether our main result on diversification holds for this distribution.
It turns out that the conclusion in Theorem \ref{th:casen} does not hold in general for the St.~Petersburg lottery, but we have a weaker version of diversification benefit. 

\begin{proposition}\label{prop:SP}
For $n\in\N$, let  $\bm X_n $  be a vector of $n$  iid St.~Petersburg lotteries. For $k,\ell\in \N$, if $\ell/k=2^a$ for some $a\in \N$, then $$\bm \gamma_k \cdot \bm X_k \le_{\rm st} \bm \gamma_\ell \cdot \bm X_\ell.$$  
\end{proposition}
\begin{proof}
We first show 
$X_1 \le_{\rm st}(X_1+X_2)/2$ 
for two iid St.~Petersburg lotteries $X_1$ and $X_2$. Note that for $m\in \N$,
\begin{align*}
\p\(\frac{X_1+X_2}{2}<2^m\)&=\p\(X_1 \le 2^m, X_2\le  2^m\) - \p\(X_1 = 2^m, X_2=  2^m\)\\
&=(1-2^{-m})^2 - (2^{-m})^2=1-2^{1-m} = \P\(X_1<2^m\). 
\end{align*}
Moreover, for any $x\in(2^{m-1},2^{m}]$, 
$$\p\(0.5(X_1+X_2)< x\)\le \p\(0.5(X_1+X_2)< 2^m\)=\P(X_1< x).$$ 
Therefore, we conclude   $X_1\le_{\rm st}(X_1+X_2)/2$. 
As the stochastic order $\ge_{\rm st}$ is closed under convolution,  $\bm \gamma_k \cdot \bm X_k \le_{\rm st} \bm \gamma_{2k} \cdot \bm X_{2k}.$
 The result holds by induction. 
\end{proof}
The stochastic dominance relation 
$\bm \gamma_k \cdot \bm X_k \le_{\rm st} \bm \gamma_\ell \cdot \bm X_\ell$ with $k\le \ell$
does not necessarily hold for $k,\ell\in\N$ if $\ell/k$ is not a power of $2$, as shown by the example below.
\begin{example}
Let $X_1, X_2, X_3$ be iid St.~Petersburg lotteries. Then 
\begin{align*}
    \p\(\frac{X_1 + X_2 + X_3}{3} <2^3\)&=\sum_{n=1}^42^{-n}\p\(X_2+X_3<24-2^n\)\\
    &=\sum_{n=1}^42^{-n}\sum_{m=1}^42^{-m}\p\(X_3<24-2^n-2^m\)\\
    &= \frac{1}{2} \times \frac{55}{64} + \frac{1}{4}\times  \frac{53}{64} + \frac{1}{8}  \times   \frac{24}{32} + \frac{1}{16}  \times \frac{1}{2} \\
&= \frac{195}{256} >0.75=\P(X_1 < 2^3).
\end{align*}
Thus, $X_1\le_{\rm st} (X_1+X_2+X_3)/3$ does not hold.
\end{example}

\section{Extensions of Theorem \ref{th:casen}}\label{sec:3}

In this section, to further understand how diversification works for extremely heavy-tailed Pareto random variables, we generalize Theorem \ref{th:casen} to several models. Specifically, in Section \ref{sec:3.1}, we study Pareto random variables caused by events with a low probability of occurrence. Section \ref{sec:3.2} considers non-negative random variables whose tail region follows a Pareto distribution. A result regarding Pareto distributions with an upper bound is given in Section \ref{sec:3.3}. Finally, Section \ref{sec:3.4} examines the case of positively dependent Pareto random variables using a stochastic representation of Pareto distributions. 

\subsection{Triggering event}\label{sec:3.1}

Extremely heavy-tailed phenomenons are often triggered by events with very small probabilities of occurrence. For instance, profits from technology breakthroughs and losses caused by severe catastrophes can be extremely heavy-tailed. In this context, it is natural to model the outcome of a rare event as $X\id_A$ where $X$ is a Pareto$(\alpha)$ random variable and $A$ is the triggering event independent of $X$ (i.e., given the occurrence of event $A$, the outcome has a Pareto$(\alpha)$ distribution). Let $X_1,\dots,X_n$ be  iid Pareto$(\alpha)$ random variables with $\alpha\in(0,1]$, and $A_1,\dots,A_n$ be the respective triggering events of  $X_1,\dots,X_n$ such that $A_1,\dots,A_n$ are independent of  $X_1,\dots,X_n$. Let $\bm Y= (X_1\id_{A_1},\dots,X_n\id_{A_n})$.

If $A_1=\dots=A_n$, then $X_1,\dots,X_n$ represent different outcomes caused by the same triggering event. For $\bm\th, \bm\eta\in \R_+^n$ satisfying $\bm\th \preceq \bm\eta$, we obtain from Theorem \ref{th:casen} the following stochastic inequality:
\begin{equation*}
     \bm\th\cdot \bm Y \ge_{\rm st} \bm\eta \cdot \bm Y.
\end{equation*}
The theorem below shows that the above inequality also holds for any events $A_1,\dots,A_n$ with an equal probability of occurrence and an arbitrary dependence. 

\begin{theorem}
 \label{th:cata}
Let  $X_1,\dots,X_n$ be iid $\mathrm{Pareto}(\alpha) $ random variables with $\alpha\in(0,1]$, and $A_1,\dots,A_n$ be events with equal probability, which are independent of $X_1,\dots,X_n$. Suppose that $\bm Y= (X_1\id_{A_1}$, $\dots, X_n\id_{A_n})$. For $\bm\th, \bm\eta\in \R_+^n$ satisfying $\bm\th \preceq \bm\eta$, we have $\bm\th\cdot \bm Y \ge_{\rm st} \bm\eta\cdot \bm Y$.
\end{theorem}

\begin{proof}
We assume $\|\bm\th\|=\|\bm\eta\|=1$ without losing generality. Below, we first show the case that $n=2$. For $\lm\in(0,1/2]$, let $S(\lm)=\p\(\lm Y_1+(1-\lm)Y_2>x\)$, where $x\in \R$. It suffices to show that $S(\lambda)$ is increasing in $\lm\in(0,1/2]$. Let $\p(A_1)=\p(A_2)=p\in(0,1)$. If $x<0$, $S(\lambda)=1$ for all $\lm\in(0,1/2]$. If $0\le x<\lm$, $S(\lm)=1-\p(A_1^c\cap A_2^c)$. If $\lm\le x\le 1-\lm$, we have
\begin{eqnarray*}
    S(\lm) &=& \p\(\{\lm X_1>x\}\cap\{A_1\cap A_2^c\}\)+\p\(A_1^c\cap A_2\)+\p\(A_1\cap A_2\)\\
           &=& \frac{\lm^\alpha}{x^\alpha}\p\(A_1\cap A_2^c\)+\p\(A_2\).
\end{eqnarray*}
As $\p(A_1)=\p(A_2)$, $\p(A_1\cap A_2^c)=\p(A_1^c\cap A_2)$. Then, if $1-\lm< x<1$, we have
\begin{eqnarray*}
     S(\lm)&=&\p(\{\lm X_1>x\}\cap\{A_1\cap A_2^c\})+\p(\{(1-\lm) X_2>x\}\cap\{A_1^c\cap A_2\})+\p(A_1\cap A_2)\\
           &=&\frac{\lm^\alpha+(1-\lm)^\alpha}{x^\alpha}\p(A_1\cap A_2^c)+\p(A_1\cap A_2).
\end{eqnarray*}
Take $a,b$ such that $0< a\le b \le 1/2$. Using the fact that $\alpha\le 1$, it is clear that $S(a)\le S(b)$ if $x\in[0,1-b]\cup(1-a,1)$. If $1-b\le x\le 1-a$,
\begin{eqnarray*}
     S(b)-S(a) &=& \frac{\p(A_1\cap A_2^c)}{x^\alpha} \big [b^\alpha+(1-b)^\alpha \big ]+\p(A_1\cap A_2)-\frac{\p(A_1\cap A_2^c)}{x^\alpha}a^\alpha-\p(A_2)\\
        &=& \frac{\p(A_1\cap A_2^c)}{x^\alpha} \big [b^\alpha+(1-b)^\alpha-a^\alpha-x^\alpha \big ]\\
        &\ge & \frac{\p(A_1\cap A_2^c)}{x^\alpha}\big [b^\alpha+(1-b)^\alpha-a^\alpha-(1-a)^\alpha \big ]\ge 0.
\end{eqnarray*}
Hence, $S(\lm)$ is increasing on $(0,1/2]$ if $0\le x<1$. If $x\ge1$, we have
\begin{eqnarray*}
    S(\lm)&=&\p(\{\lm X_1>x\}\cap\{A_1\cap A_2^c\})+\p(\{(1-\lm) X_2>x\}\cap\{A_1^c\cap A_2\})\\
     & & +\p(\{\lm X_1+(1-\lm) X_2>x\}\cap\{A_1\cap A_2\})\\
     &=&\frac{\p(A_1\cap A_2^c)}{x^\alpha} \big [\lm^\alpha+(1-\lm)^\alpha\big ] +\p(A_1\cap A_2)\p\big (\lm X_1+(1-\lm) X_2>x\big ).
\end{eqnarray*}
By Lemma \ref{le-3-1}, it is clear that $S(\lambda)$ is strictly increasing on $(0,1/2]$ if $x>1$. Hence, we have the result for $n=2$.

Next, we show the result for $n\ge 3$. By Lemma \ref{lem:maj}, it suffices to show that $\bm\th\cdot \bm Y \ge_{\rm st} \bm\eta\cdot \bm Y$ holds for $\bm\th$ being a $T$-transform of $\bm\eta$.  Write  $\bm\th=(\th_1,\dots, \th_n)$ and $\bm\eta =(\eta_1,\dots,\eta_n)$. Take $k,\ell \in[n]$ such that $k\ne \ell$ and $(\th_k,\th_\ell)$ is a $T$-transform of $(\eta_k,\eta_\ell)$. For $S\subseteq [n]$, let $B_S=(\bigcap_{i\in S }A_i)\cap(\bigcap_{i\in S^c }A_i^c)$. For $(\theta_1,\dots,\theta_n)\in \R_+^n$, we write
\begin{align*}
    \sum_{i=1}^n\theta_iX_i\id_{A_i}=&\sum_{S\subseteq[n]/\{k,\ell\}}\id_{B_S}\sum_{i\in S}\theta_iX_i+\sum_{\{k,\ell\}\subseteq S\subseteq[n]}\id_{B_S}\sum_{i\in S}\theta_iX_i\\
    & +\sum_{\{k\}\subseteq S\subseteq[n]/\{\ell\}}\id_{B_S}\sum_{i\in S}\theta_iX_i+\sum_{\{\ell\}\subseteq S\subseteq[n]/\{k\}}\id_{B_S}\sum_{i\in S}\theta_iX_i.
\end{align*}
It is clear that
$$
    \sum_{S\subseteq[n]/\{k,\ell\}}\id_{B_S}\sum_{i\in  S} \th_i X_i = \sum_{S\subseteq[n]/\{k,\ell\}}\id_{B_S}\sum_{i\in S}\eta_iX_i.
$$
By Theorem \ref{th:casen}, we have
$$
    \sum_{\{k,\ell\}\subseteq S\subseteq[n]}\id_{B_S}\sum_{i\in S}\theta_iX_i\ge_{\rm st} \sum_{\{k,\ell\}\subseteq S\subseteq[n]}\id_{B_S}\sum_{i\in S}\eta_iX_i.
$$
Note that
\begin{align*}
     \sum_{\{k\}\subseteq S\subseteq[n]/\{\ell\}}\id_{B_S}\sum_{i\in S}\theta_iX_i
           =\sum_{ D\subseteq[n]/\{k,\ell\}}\id_{A_k}\id_{A_\ell^c}\prod_{s\in D}\id_{A_s}\prod_{t\in D^c}\id_{A_t}\(\theta_kX_k+\sum_{i\in D}\theta_iX_i\).
\end{align*}
Then
\begin{align*}
   & \sum_{\{k\}\subseteq S\subseteq[n]/\{\ell\}}\id_{B_S}\sum_{i\in S}\theta_iX_i+\sum_{\{\ell\}\subseteq S\subseteq[n]/\{k\}}\id_{B_S}\sum_{i\in S}\theta_iX_i\\
   & \qquad = \sum_{ D\subseteq[n]/\{k,\ell\}}\prod_{s\in D}\id_{A_s}\prod_{t\in D^c}
      \id_{A_t}\(\id_{A_k}\id_{A_\ell^c}\bigg (\theta_kX_k+\sum_{i\in D}\theta_iX_i\bigg )
      +\id_{A_k^c}\id_{A_\ell}\bigg (\theta_\ell X_\ell+\sum_{i\in D}\theta_iX_i\bigg ) \).
\end{align*}
For $D\subseteq[n]/\{k,\ell\}$, let $f$ denote the density of $\sum_{i\in D}\theta_iX_i$. For $s>0$, we have
\begin{align*}
    &\p\(\id_{A_k}\id_{A_\ell^c}\bigg (\theta_kX_k+\sum_{i\in D}\theta_iX_i\bigg )
            +\id_{A_k^c}\id_{A_\ell}\bigg (\th_\ell X_\ell+\sum_{i\in D}\th_i X_i\bigg )>s\)\\
      & \quad = \p(A_k\cap A_\ell^c)\(\p\bigg ( \theta_kX_k+\sum_{i\in D}\theta_iX_i>s\bigg )
             +\p\bigg (\theta_\ell X_\ell+\sum_{i\in D}\theta_iX_i>s\bigg ) \)\\
      &\quad = \p(A_k\cap A_\ell^c)\(\int_{-\infty}^\infty\p(\theta_kX_k>s-t)f(t)\,\d t
              +\int_{-\infty}^\infty\p(\theta_\ell X_\ell>s-t)f(t)\,\d t\)\\
      &\quad = \p(A_k\cap A_\ell^c)\int_{-\infty}^\infty\(\p(\theta_k X_k>s-t)
                +\p(\theta_\ell X_\ell>s-t)\)f(t)\,\d t\\
      & \quad \ge \p(A_k\cap A_\ell^c)\int_{-\infty}^\infty\(\p(\eta_kX_k>s-t)
             +\p(\eta_\ell X_\ell>s-t)\)f(t)\, \d t.
\end{align*}
The above inequality is obtained by applying the result for $n=2$ on mutually exclusive events. Hence, we have the result for $n\ge 3$.
\end{proof}

\subsection{Pareto tails}\label{sec:3.2}

In practice, random variables may not follow Pareto distributions in their entire support, whereas they have power-like distributions beyond some high thresholds  (implied by the Pickands-Balkema-de Haan Theorem; see, e.g., Theorem 3.4.13 (b) in \cite{EKM97}). Therefore, it is practically useful to assume that a random variable has a Pareto distribution only in the tail region. For $\alpha>0$, we say that $Y$ has a {Pareto}($\alpha$) distribution beyond $c\ge 1$  if  $\p(Y>t)=t^{-\alpha}$ for $t\ge c$ and $\p(Y>0)=1$.                                              

\begin{proposition}
 \label{prop:tail}
Let $\bm Y =(Y_1,\dots,Y_n)$ be a vector of $n$ iid $\mathrm{Pareto}(\alpha) $ random variables beyond $c\ge 1$ with $\alpha\in(0,1]$.  For $\bm\th, \bm\eta\in \R_+^n$ satisfying $\bm\th\prec \bm\eta$, we have $\p(\bm\th\cdot \bm Y>x)>\p(\bm\eta\cdot \bm Y>x)$ for $x>c\|\bm\th\|$.
\end{proposition}

\begin{proof}
We assume that $\|\bm\th\|=\|\bm\eta\|=1$ without losing generality. Below, we show the case that $n=2$. For $\lm\in(0,1/2]$, let $S(\lm)=\p\(\lm Y_1+(1-\lm)Y_2>x\)$ where $x>c$. It suffices to show that $S(\lambda)$ is increasing as $\lm$ increases. Denote by $\mu$ the probability measure of $Y_1$ and $Y_2$. We have
\begin{align*}
   S(\lm) &=\p\(\lm Y_1+(1-\lm)Y_2>x\)\\
    & =\p\(\lm Y_1+(1-\lm)Y_2>x, Y_2\le c\)+\p\(\lm Y_1+(1-\lm)Y_2>x,Y_1\le c\)\\
    & \qquad +\p\(\lm Y_1+(1-\lm)Y_2>x,Y_1>c, Y_2> c\)\\
    & =\int_0^c \p\(Y_1>\frac{x-(1-\lm)y}{\lm}\)\,\d\mu(y)+\int_0^c \p\(Y_2>\frac{x-\lm y}{1-\lm}\) \d\mu(y)\\
    & \qquad +\p\(\lm Y_1+(1-\lm)Y_2>x,Y_1>c, Y_2> c\)\\
    & =\int_{0}^c \(\(\frac{\lm}{x-(1-\lm)y}\)^\alpha+\(\frac{1-\lm}{x-\lm y}\)^\alpha\)\d \mu(y)
            +\p\(\lm Y_1+(1-\lm)Y_2>x,Y_1>c, Y_2> c\).
 \end{align*}
For $\lm\in(0,1/2]$, let
\begin{align*}
      S_1(\lm) & =\int_{-\infty}^c \(\(\frac{\lm}{x-(1-\lm)y}\)^\alpha+\(\frac{1-\lm}{x-\lm y}\)^\alpha\)\d\mu(y);\\
  S_2(\lm) &=\p\(\lm Y_1+(1-\lm)Y_2>x, Y_1>c, Y_2> c\).
\end{align*}
Fix $x>c$ and $y\le c$. For $t\in (0,1)$, let
$$     g(t)=\(\frac{t}{x-(1-t)y}\)^\alpha.        $$
Taking the second derivative of $g$, we have
$$
        g''(t)=\frac{\alpha(x-y)x^\alpha(x-(1-t)y)^{-\alpha}((\alpha-1)(x-y)-2ty)}{t^2(x-(1-t)y)^2}\le0.
$$
Therefore, $g$ is concave, which implies that $g(\lm)+g(1-\lm)$ is increasing on $(0,1/2]$. Hence, $S_1(\lm)$ is increasing on $(0,1/2]$.

Let $X_1, X_2$ be iid $\mathrm{Pareto}(\alpha)$ random variables. Note that, for $t\ge c>1$, $\p(Y_1>t|Y_1>c)=(c/t)^\alpha=\p(c X_1>t)$. Hence, we have
$$
      \p\(\lm Y_1+(1-\lm)Y_2>x|Y_1>c, Y_2> c\)=\p\(\lm X_1+(1-\lm) X_2>\frac{x}{c}\).
$$
By Lemma \ref{le-3-1}, $S_2(\lambda)$ is strictly increasing on $(0,1/2]$. Hence, the result for $n=2$ is shown. Next, by a similar proof as that of Theorem \ref{th:casen}, we obtain the desired result for $n>2$.
\end{proof}

In the realm of insurance, excess-of-loss reinsurance is frequently used, i.e., only losses beyond some high threshold are covered. Proposition \ref{prop:tail} leads to the following corollary.

\begin{corollary}
   \label{cor:tail}
Let $X_1,\dots,X_n$ be iid $\mathrm{Pareto}(\alpha) $ random variables with $\alpha\in(0,1]$, and  $\bm\th, \bm\eta\in \R_+^n$ satisfy $\bm\th\preceq \bm\eta$.
\begin{enumerate}[(i)]
  \item Let $\bm Y= (X_1\vee c,\dots,X_n\vee c)$ where $c\ge 1$. Then
   $\bm\th\cdot \bm Y \ge_{\rm st} \bm\eta\cdot \bm Y$. Moreover, if $\bm\th\prec \bm\eta$, then $\bm\th\cdot \bm Y >_{\rm st} \bm\eta\cdot \bm Y$.

  \item Let $\bm Y= ((X_1-c)_+,\dots,(X_n-c)_+)$ where $c\ge 1$. Then $\bm\th\cdot \bm Y \ge_{\rm st} \bm\eta \cdot \bm Y$. Moreover, if $\bm\th \prec \bm\eta$, then $\bm\th\cdot \bm Y >_{\rm st} \bm\eta\cdot \bm Y$.
\end{enumerate}
\end{corollary}

\begin{proof}
The statement (i) is a special case of Proposition \ref{prop:tail}, and (ii) is obtained from (i) by shifting locations of the random variables in (i).
\end{proof}

For the end of this subsection, we have the following Proposition \ref{non} regarding the case of different marginals. 

\begin{proposition}
\label{non}
Let $\bm X = (X_1, \dots, X_n)$ be a vector with independent components such that $X_i \sim {\rm Pareto}(\alpha_i)$ for $i = 1, \dots, n$, where $\alpha_i \in (0,1]$ such that $\alpha_1 \le \dots \le \alpha_n$. Consider $\bm Y= ((X_1-c)_+,\dots,(X_n-c)_+)$ where $c\ge 1$, and  $\bm\th, \bm\eta\in \R_+^n$ satisfying $\boldsymbol{\theta} \preceq \bm\eta$. Then $\bm\th^\uparrow \cdot \bm Y \ge_{\rm st} \bm\eta^\uparrow \cdot \bm Y$. Moreover, if $\bm\th \prec \bm\eta$, then $\p(\bm\th^\uparrow \cdot \bm Y>x)>\p(\bm\eta^\uparrow \cdot \bm Y>x)$ for $x>0$.
\end{proposition}

\begin{proof}
We assume that $\|\bm\th\|=\|\bm\eta\|=1$ without losing generality. Below, we show the case that $n=2$. For $\lm\in(0,1/2]$, let $S(\lm)=\p\(\lm Y_1+(1-\lm)Y_2>x\)$ where $x\in \R$. It suffices to show that $S(\lambda)$ is increasing as $\lm$ increases. If $x<0$, $S(\lm)=1$ for all $\lm\in(0,1/2]$. If $x\ge0$,   we have
\begin{align*}
	S(\lm) &= \p\(\lm Y_1+(1-\lm)Y_2>x\)\\
	&= \p\(\lm Y_1>x, X_2\le c\)+\p\((1-\lm)Y_2>x,X_1\le c\) +\p\(\lm Y_1+(1-\lm)Y_2>x,X_1>c, X_2> c\)\\
  &=\(\frac {x}{\lm} +c\)^{-\alpha_1}(1-c^{-\alpha_2})+\(\frac {x}{1-\lm} +c\)^{-\alpha_2}(1-c^{-\alpha_1})\\
    & \quad  +\p\(\lm X_1+(1-\lm)X_2>x+c,X_1>c, X_2> c\).
\end{align*}
For $\lm\in(0,1/2]$, let
\begin{align*}
	S_1(\lambda)&= \(\frac {x}{\lm}+c\)^{-\alpha_1}(1-c^{-\alpha_2})
                +\(\frac {x}{1-\lm}+c\)^{-\alpha_2}(1-c^{-\alpha_1});\\
	S_2(\lambda)&= \p\(\lm X_1+(1-\lm)X_2>x+c,X_1>c, X_2> c\).
\end{align*}
Taking the derivative of $S_1(\lm)$, and noting that for $\alpha_1 \le \alpha_2$ and $\lm\in(0,1/2)$, we have
\begin{align*}
  S_1'(\lm) &= x\(\alpha_1\lm^{\alpha_1-1}\(x+c\lm\)^{-\alpha_1-1}(1\!-\! c^{-\alpha_2})
        - \alpha_2(1-\lm)^{\alpha_2-1} \(x+c(1-\lm)\)^{-\alpha_2-1} (1\!-\! c^{-\alpha_1})\) \\
    &\ge  x c^{-\alpha_2} (1-\lm)^{\alpha_1-1} \(x+c(1-\lm)\)^{-\alpha_2-1} \\
     & \quad \times \(\alpha_1 \(x+c(1-\lm)\)^{\alpha_2-\alpha_1} (c^{\alpha_2}-1)-\alpha_2\(c(1-\lm)\)^{\alpha_2- \alpha_1} (c^{\alpha_1}-1) \) \\
    & = xc^{-\alpha_2} (1-\lm)^{\alpha_1-1}\(x+c(1-\lm)\)^{-\alpha_2-1}\(c(1-\lm)\)^{\alpha_2- \alpha_1} (c^{\alpha_1}-1) \alpha_1 \\
    & \quad \times \(\(\frac{x+c(1-\lm)}{c(1-\lm)}\)^{\alpha_2-\alpha_1} \frac{c^{\alpha_2}-1}{c^{\alpha_1}-1} - \frac{\alpha_2}{\alpha_1}  \) \\
    & \ge x c^{-\alpha_2} (1\!-\!\lm)^{\alpha_1-1} \(x+c(1\!-\!\lm)\)^{-\alpha_2-1}
           \(c(1\!-\!\lm)\)^{\alpha_2- \alpha_1}  \frac{\alpha_1 \alpha_2}{c^{\alpha_1}-1} \(\frac{c^{\alpha_2}-1}{\alpha_2}-\frac{c^{\alpha_1}-1}{\alpha_1}\)\\
     &\ge 0,
\end{align*}
where the last inequality follows from the fact that the function $h(x):= (c^x-1)/x$ is increasing in $x\in (0,1]$. Therefore, $S_1(\lambda)$ is strictly increasing on $(0,1/2]$.
	
Note that, for $t\ge c>1$ and $i=1,2$, $\p(X_i>t|X_i>c)=(c/t)^{\alpha_i}=\p(cX_i>t)$. Hence, we have
$$   \p\(\lm X_1+(1-\lm)X_2>x+c|X_1>c, X_2> c\)=\p\(\lm X_1+(1-\lm)X_2>\frac{x}{c}+1\).$$
By Lemma \ref{le-3-1}, $S_2(\lambda)$ is also strictly increasing on $(0,1/2]$. Hence, the result for $n=2$ is shown. Next, by a similar proof to that of Theorem \ref{th:casen}, we obtain the desired result for $n>2$.
\end{proof}

\subsection{Bounded Pareto random variables}\label{sec:3.3}

Next, we consider random variables modelled by extremely heavy-tailed Pareto distributions with an upper bound. For $c\in(1,\infty)$, let $\bm Y=(X_1\wedge c,\dots,X_n\wedge c)$, where $X_1,\dots,X_n$ are iid Pareto$(\alpha)$ random variables with $\alpha\le 1$. For  $\boldsymbol{\theta},\boldsymbol{\eta}\in \R_+^n$ satisfying $\boldsymbol{\theta} \preceq \boldsymbol{\eta} $,  by Proposition \ref{prop:infinite}, $\boldsymbol{\theta} \cdot \bm Y \ge_{\rm st} \boldsymbol{\eta} \cdot \bm Y$ does not hold. However, the following proposition shows that the tail probability of $ \boldsymbol{\theta} \cdot \bm Y$ still dominates that of $ \boldsymbol{\eta} \cdot \bm Y$ if the upper bound is large enough.
 \begin{proposition}\label{prop:bounded}
Fix $\boldsymbol{\theta},\boldsymbol{\eta}\in \R_+^n$ such that $\boldsymbol{\theta} \preceq \boldsymbol{\eta}$ and $\boldsymbol{\eta}=(\eta_1,\dots,\eta_n)$ with $\eta_{(1)} := \min (\eta_1,\dots,\eta_n) > 0$. Let $b=\|\bm \eta\|/\eta_{(1)}$.  Let  $X_1,\dots,X_n$ be iid $\mathrm{Pareto}(\alpha) $ random variables with $\alpha\in(0,1]$ and $\bm Y=(X_1\wedge c,\dots,X_n\wedge c)$ with $c\in(b,\infty)$. Then
$\p(\boldsymbol{\theta} \cdot \bm Y>x)\ge\p(\boldsymbol{\eta} \cdot \bm Y>x)$ for $x\in(\|\boldsymbol{\eta}\|,(c/b)\|\boldsymbol{\eta}\|)$.
\end{proposition}
\begin{proof}
Let $\bm X =(X_1,\dots,X_n)$. If there is at least one $X_i > c$ with $i \in [n]$, then $\sum_{i=1}^n\eta_{i}(X_i\wedge c) \geq \sum_{j \neq i} \eta_{j} + \eta_i c \geq \eta_{(1)} c = (c/b)\|\boldsymbol{\eta}\|$. Thus, for $x\in(\|\boldsymbol{\eta}\|,(c/b)\|\boldsymbol{\eta}\|)$,we have
\begin{align*}
\p\(\bm{\eta}\cdot\bm{Y}\le  x\) &= \p\(\sum_{i=1}^n\eta_{i}(X_i\wedge c)\le  x\) \\
								 &= \p\(\sum_{i=1}^n\eta_{i} X_i \le  x, X_i \leq c, i \in [n]\) = \p\(\bm{\eta}\cdot\bm{X}\le  x\).
\end{align*}
Let $\boldsymbol{\theta}=(\theta_1,\dots,\theta_n)$. As $\boldsymbol{\theta} \preceq \boldsymbol{\eta},\ (c/b)\|\boldsymbol{\eta}\|\le c \min(\theta_1, \dots, \theta_n).$ Hence, we also have $\p\(\bm{\theta}\cdot\bm{Y}\le  x\)=\p\(\bm{\theta}\cdot\bm{X}\le  x\)$. By Theorem \ref{th:casen}, we have the desired result.
\end{proof} 
Proposition \ref{prop:bounded} implies that if the upper bound is large enough (e.g., the total wealth in an economy), a more diversified portfolio of extremely heavy-tailed Pareto random variables can dominate a less diversified one in the sense of tail probability in a large region.

\subsection{Positive dependence}\label{sec:3.4}

So far, all results require independence among the components of $\mathbf X$. 
\cite{CEW24} obtained results also under a notion of negative dependence. Our techniques for proving Theorem \ref{th:casen} do not generalize to negative dependence, but we are able to obtain some results under a specific form of positive dependence.
 See, e.g., \cite{PW15} for concepts of positive and negative dependence, including comonotonicity, which we will mention in passing below.

We   consider the case that Pareto random variables are positively dependent via a common shock.
Let $\mathrm{Gamma}(\alpha)$, $\alpha>0$, denote the distribution function of a gamma random variable with density $f(x)=x^{\alpha-1}\exp(-x)/\Gamma(\alpha)$ for $x>0$.
Note that a Pareto random variable $X\sim \mathrm{Pareto}(\alpha)$, $\alpha>0$, has the following stochastic representation 
$$X=\frac{Z_1+Z}{Z},$$ 
where $Z_1\sim \mathrm{Gamma}(1)$ and $Z\sim \mathrm{Gamma}(\alpha)$ are independent (see, e.g., Lemma 1 of \cite{SGPJ16}). 

Let $Z_1,\dots,Z_n\sim \mathrm{Gamma}(1)$ and $Z\sim \mathrm{Gamma}(\alpha)$, $\alpha>0$, be independent. Consider a multivariate Pareto distribution with parameter $\alpha>0$, denoted by $\mathrm{MP}(\alpha,n)$, of which the associated random vector has the stochastic representation below
\begin{equation}\label{eq:MP}
\(X_1,\dots,X_n\)=\(\frac{Z_1+Z}{Z},\dots,\frac{Z_n+Z}{Z}\).\end{equation} 
Clearly, $X_1,\dots,X_n$ are positively dependent as they are all affected by the common shock $Z$; if $\alpha>2$, $\rho(X_i,X_j)=1/\alpha$ for $i\neq j$ (see p.~155 of \cite{SGPJ16}). The distribution of $(X_1-1,\dots,X_n-1)$ is the Type II multivariate Pareto distribution; see \cite{A15}. 
The survival copula
\footnote{Joint distributions of uniform random variables over $(0,1)$ are known as copulas; see \cite{N06} for an introduction.}
associated with $(X_1,\dots,X_n)$ is the Clayton copula with parameter $1/\alpha$, given by 
$$C(u_1,\dots,u_n)=\(u_1^{-1/\alpha}+\dots+u_n^{-1/\alpha}-n+1\)^{-\alpha},~~(u_1,\dots,u_n)\in(0,1)^n.$$
The Clayton copula approaches   independence as $\alpha$ goes to  $\infty$ and it approaches   comonotonicity as $\alpha$ goes to $0$.

\begin{theorem}\label{thm:Clayton}
Let $\bm X \sim \mathrm{MP}(\alpha,n)$ with  $\alpha \in (0,1]$ and $n\in \N$. For $\bm\th, \bm\eta\in \R_+^n$ satisfying $\bm\th \preceq \bm\eta$, we have $\bm{\eta} \cdot\bm X \le_{{\rm st}} \bm{\theta} \cdot \bm X$.
\end{theorem}

\begin{remark}[Feller-Pareto distribution]
Theorem 3 also holds for a class of multivariate Feller-Pareto distributions, which  includes the Type II multivariate Pareto distribution as a special case (see Theorem \ref{thm:MFP}). We provide a proof in the appendix for Theorem \ref{thm:MFP}  which implies Theorem \ref{thm:Clayton}. 
\end{remark}

\begin{remark}[Mixture dependence]
Let the components of $\bm X=(X_1,\dots,X_n)$ be identically distributed Pareto random variables with tail parameter $\alpha\le 1$. 
It is clear that the inequality $\bm\th\cdot \bm X \ge_{\rm st} \bm\eta\cdot\bm X$ holds if $X_1,\dots,X_n$ are comonotonic (i.e., perfectly positively dependent). Therefore, by Theorems \ref{th:casen} and \ref{thm:Clayton}, $\bm\th\cdot \bm X \ge_{\rm st} \bm\eta\cdot\bm X$ also holds if the dependence structure (i.e., copula) of $\bm X$ is a mixture of independence, comonotonicity, and the Clayton copula with parameter $1/\alpha$.
\end{remark}

 \section{Investment implications}\label{sec:4}
Suppose that an agent needs to decide on several investments whose profits are iid Pareto random variables. Let $\X$ be the set of all random variables.  The agent is equipped with a \emph{preference function} $\rho: \mathcal X_\rho\rightarrow \overline \R:= [-\infty,\infty]$ where $\X_\rho\subseteq \X$ is a set of random variables representing profits. The agent prefers $Y$ over $X$ if and only if $\rho(X) \le \rho(Y)$.  For the rest of our discussions, we make minimal assumptions of monotonicity on $\rho $ in the following two forms.
\begin{enumerate}[(a)]
 \item Weak monotonicity: $\rho(X) \le \rho(Y)$ for $X,Y\in\mathcal X_\rho$ if $X\le_{\rm st} Y$.
 \item Mild monotonicity: $\rho$ is weakly monotone and $\rho(X)< \rho(Y)$ if $\p(X<Y)=1$.
\end{enumerate}

The two assumptions are satisfied by all commonly used preference models. As a classic example, the expected utility model in decision analysis, denoted by $E_u$, is defined as
$$
E_u(X) = \E[u(X)],~~~~~~X\in \X_{E_u}:=\{Y\in \X: \E[u(Y)]\mbox{~is well-defined}\},
$$
 where the utility function $u$ is  measurable,  and typically assumed to be increasing.
 Here, for $\E[u(Y)]$ to be well-defined, it can be infinite. 
 Clearly, $E_u$ is weakly monotone and is mildly monotone if $u$ is strictly increasing.
An expected utility agent is \emph{risk averse} if  $u$ is concave, in the sense of \cite{RS70}.  The following result shows that risk-averse expected utility agents always prefer a more diversified portfolio of iid Pareto profits, regardless of the tail parameter $\alpha$.
\begin{proposition}
    Let $\bm X$ be a vector of $n$ iid Pareto random variables and $u$ be an increasing and   concave function. 
    For $\bm\th, \bm\eta\in \R_+^n$ satisfying $\bm\th \preceq \bm\eta$, we have $E_u(\bm{\theta} \cdot \bm X) \ge E_u(\bm{\eta} \cdot\bm X)$. 
\end{proposition}
\begin{proof}
Let $\alpha$ be the tail parameter of the iid Pareto random variables.
 As $\mathbf X$ has nonnegative components,  $E_u(\bm{\theta} \cdot \bm X)$ and $E_u(\bm{\eta} \cdot\bm X)$ are always well-defined.  If $\alpha\le 1$, the result directly follows from Theorem \ref{th:casen} as $E_u$ is weakly monotone. If $\alpha>1$, we have $\bm \eta\cdot \bm X\le_{\rm ssd}\bm \theta\cdot \bm X$ by Proposition \ref{prop:finite}. Through the definition of second-order stochastic dominance, as $u$ is increasing and concave, we have the desired result.
\end{proof}

For an expected utility decision maker,  we can observe a sharp contrast between finite-mean and infinite-mean models. 
Consider a decision maker with preference function $\rho$.
For a random vector $\mathbf X$, we say that \emph{pro-diversification} holds 
if $\rho(\bm{\theta} \cdot \bm X) \ge \rho(\bm{\eta} \cdot\bm X)$
 for all $\bm\th, \bm\eta\in \R_+^n$ satisfying $\bm\th \preceq \bm\eta$. Intuitively, 
 this means that the decision maker likes more diversification. 
\begin{theorem}\label{th:RW}
    Consider an expected utility  decision maker with a utility function $u$.
    \begin{enumerate}[(i)]
    \item Pro-diversification holds for all iid $\mathbf X$ with finite mean if and only if $u$ is concave.  
    \item Pro-diversification holds for all iid $\mathbf X$ with two-point marginals if and only if $u$ is concave. 
    \item Pro-diversification holds for all iid $\mathbf X$ with    $\mathrm{Pareto}(\alpha)$ marginals for $\alpha\in (0,1] $ if $u$ is increasing. 
    \end{enumerate}
\end{theorem}
\begin{proof}
 If $u$ is concave, then by  Theorem 3.A.35 of \cite{SS07}, 
$E_u(\bm{\theta} \cdot \bm X) \ge E_u(\bm{\eta} \cdot\bm X)$ if $\bm\th \preceq \bm\eta$ for  all iid $\mathbf X$ with finite mean, and thus pro-diversification holds.
This shows the ``if" direction of (i).
Suppose that pro-diversification holds for $(X_1,X_2)$, where each of $X_1$ and $X_2$ takes values $a,b\in \R$ with $1/2$ probability.
We have $ E_u(X_1/2+X_2/2) \ge  E_u(X_1)$, yielding $u((a+b)/2)\ge u(a)/2+ u(b)/2$, which implies concavity by  Sierpi\' nski's Theorem on mid-point concavity.
This shows the ``only if" direction of (ii). The ``only if" direction of (i) and the ``if" direction of (ii) are implied by the above two shown directions. 
Statement (iii) follows directly from Theorem \ref{th:casen}.
\end{proof}
Theorem \ref{th:RW} clarifies the fact that pro-diversification for finite-mean models is a result of risk aversion, but 
for infinite-mean Pareto models it is not related to risk aversion. As far as we are aware, this observation is new to the literature. Infinite-mean Pareto distributions can be used to model profits from technological innovations (\cite{scherer2001uncertainty} and \cite{silverberg2007size}).\footnote{Infinite-mean models are not uncommon in modelling profits/returns. For instance, the log-returns of stocks and other financial assets are often modelled by heavy-tailed distributions. This implies that the return distributions must have infinite mean, although in this case,  they do not necessarily have Pareto tails.}
Our result can explain the widely observed diversified investment in many start-up firms by venture capital firms, which are, intuitively, not necessarily risk averse.

Other useful preference models are risk measures in the sense of \cite{ADEH99} and \cite{FS16}. Almost all distortion risk measures are mildly monotone; see Proposition A.1 of \cite{CEW24(2)} for the precise statement. 

The rest of this section focuses on extremely heavy-tailed Pareto random variables. Note that some preference models may give $\infty$ value in the presence of extremely heavy tailedness and are thus not useful in this case; \cite{CEW24(2)} listed some useful decision models in evaluating infinite-mean random variables.

A mapping $f$ on $\R^n$ is \emph{Schur-concave} (resp.~\emph{Schur-convex}) if $f(\bm\th)\ge f(\bm\eta) $ (resp.~$f(\bm\th)\le f(\bm\eta) $) for $\bm\th\preceq  \bm\eta$.
We speak of strict Schur-concavity by replacing  both inequalities above as strict inequalities. From now on, $\X_\rho$ is assumed to contain the convex cone of Pareto$(\alpha)$ random variables with $\alpha\le 1$. As implied by Theorem \ref{th:casen}, a more diversified portfolio of iid extremely heavy-tailed Pareto random variables is more likely to be large. The following result directly follows from Theorem \ref{th:casen}.

\begin{proposition}\label{cor:riskmeasure}
Let $\bm X$ be a vector of $n$ iid $\mathrm{Pareto}(\alpha) $ random variables with $\alpha\in(0,1]$ and $\rho:\X_\rho\to \R$. The mapping
$\bm\th\mapsto \rho(\bm\th\cdot \bm X)$ on $\R^n_+$ is Schur-concave if $\rho$ is weakly monotone, and it is strictly Schur-concave if $\rho$ is mildly monotone.
 \end{proposition}

For $i \in[n]$, let $\bm e_{i,n}$ be the $i$th column vector of the $n\times n$ identity matrix. The above result says that, for any weakly monotone risk measure $\rho$, and $\bm \th, \bm\eta\in \R_+^n$ such that $\bm\th\preceq \bm\eta$ and $\|\bm\th\|=w$,
$$
   \rho(w\bm e_{i,n}\cdot\bm X)\le \rho\(\bm\eta\cdot\bm X\)
      \le\rho\(\bm\th\cdot \bm X\) \le\rho\(\frac{w}{n} \bm 1_n\cdot \bm X\),\quad i\in[n].
$$
The above inequalities are strict if $\rho$ is mildly monotone and $w\bm e_{i,n}\prec \bm\th\prec \bm\eta\prec (w/n)\bm 1_n$.
We emphasize the importance of the strict inequalities to decision makers below.

Suppose that an agent wants to allocate their exposures onto $\bm X$ to maximize their preference. The agent faces a total position $\bm w\cdot \bm X -g(\bm w)$ where the function $g$ represents a penalty and it is Schur-convex. The agent has one of the following optimization problems:
\begin{enumerate}[(p1)]
 \item maximize\ \ $\rho\(\bm w\cdot \bm X -g(\bm w)\)$ \ \ subject to $\bm w\in\mathbb R_+^n$ and $\|\bm w\|=w$ with given $w>0$,

 \item maximize\ \ $\rho\(\bm w\cdot \bm X -g(\bm w) \)$ \ \ subject to $\bm w\in\mathbb R_+^n$.
\end{enumerate}
The following proposition follows from the strictness statement of Theorem \ref{th:casen}.

\begin{proposition}\label{prop:concentrate}
Let $\rho:\mathcal X_\rho\rightarrow  \R$ be  mildly monotone and $g:\R\to\R$. Then the maximizer to {\rm (p1)} is $w\bm {1}_n/n$, and the set of maximizers to {\rm (p2)} is in $\left\{w\bm {1}_n/n:w\in\R_+\right\}$.
\end{proposition}

\section{Conclusion}\label{sec:5}

Extending the recently obtained 
 stochastic dominance relation  \begin{equation}
 \label{eq:conc1}
      \theta X_{1}\le_{\rm st}\theta_{1}X_{1}+\dots+\theta_{n}X_{n} \mbox{~~~~for any $(\theta_1,\dots,\theta_n)\in \R_+^n$, where  $\theta =\sum_{i=1}^n\theta_n $}
\end{equation}
for iid infinite-mean Pareto random variables, this paper focuses on the stronger relation 
\begin{align}
 \label{eq:conc2}
      \eta_1 X_1+\dots+\eta_n X_n \le_{\rm st} \theta_1 X_1 +\dots+\theta_n X_n \mbox{~~~~for any $(\theta_1,\dots,\theta_n)  \preceq (\eta_1,\dots,\eta_n)\in \R_+^n$}.
\end{align}
In Theorem \ref{th:casen}, 
we show that \eqref{eq:conc2} holds for iid infinite-mean Pareto random variables. This result is generalized to several cases:  (i) Pareto random variables are caused by triggering events (Theorem \ref{th:cata}); (ii)  tail region of random variables are Pareto (Proposition \ref{prop:tail}); (iii) Pareto random variables are bounded (Proposition \ref{prop:bounded}); (iv) Pareto random variables are positively dependent (Theorem \ref{thm:Clayton}). The implications of Theorem \ref{th:casen} for monotone agents making investment decisions are discussed in Section \ref{sec:4}, showing that in our setting, preference for portfolio diversification does not necessarily require the decision maker to be risk averse, in sharp contrast to the case of finite-mean models.

Although our main result is generalized to several classes of models, all techniques rely heavily on the specific power form of the Pareto distribution functions, and they do not seem to easily generalize to other infinite-mean distributions. 

After the current paper was posted online,    some more recent developments in \cite{CS24}, \cite{M24}, and \cite{ALO24} showed \eqref{eq:conc1} for more general classes of infinite-mean  distributions than the  Pareto distributions. 
A natural question is whether \eqref{eq:conc2} also holds for these  distributions. 
An answer remains unclear to us because of the significantly different techniques employed in these papers.
In particular, we were not able to find distributions for which \eqref{eq:conc1} holds under the iid assumption but  the stronger \eqref{eq:conc2} fails to hold. Inspired by this observation, a simple conjecture is formally stated below. 
\begin{conjecture}
For iid $X_1,\dots,X_n$,
    \eqref{eq:conc1} and \eqref{eq:conc2} are equivalent. 
\end{conjecture}
Studying the above conjecture, extending  \eqref{eq:conc1} or \eqref{eq:conc2} beyond the known cases,  and fully characterizing distributions that satisfy  \eqref{eq:conc1} or \eqref{eq:conc2} require future research beyond the existing techniques.

\appendix

\setcounter{table}{0}
\setcounter{figure}{0}
\setcounter{equation}{0}
\renewcommand{\thetable}{A.\arabic{table}}
\renewcommand{\thefigure}{A.\arabic{figure}}
\renewcommand{\theequation}{A.\arabic{equation}}

\setcounter{theorem}{0}
\setcounter{proposition}{0}
\renewcommand{\thetheorem}{A.\arabic{theorem}}
\renewcommand{\theproposition}{A.\arabic{proposition}}
\setcounter{lemma}{0}
\renewcommand{\thelemma}{A.\arabic{lemma}}

\setcounter{example}{0}
\renewcommand{\theexample}{A.\arabic{example}}

\setcounter{corollary}{0}
\renewcommand{\thecorollary}{A.\arabic{corollary}}

\setcounter{remark}{0}
\renewcommand{\theremark}{A.\arabic{remark}}
\setcounter{definition}{0}
\renewcommand{\thedefinition}{A.\arabic{definition}}

\section{Appendices}
\subsection{Proof of Proposition \ref{prop:infinite}}
\begin{proof}[Proof of Proposition \ref{prop:infinite}]
 Write $\bm {\theta}=\(\theta_1,\dots,\theta_n\)$ and $\bm {\eta} = \(\eta_1,\dots,\eta_n\)$. 
 Since  $\bm\th \cdot\bm X \ge_{\rm st} \bm\eta\cdot \bm X$
and  $\E[\bm\th \cdot\bm X]=\E[ \bm\eta\cdot \bm X]$, we have that $\bm\th \cdot\bm X  $ and $\bm\eta\cdot \bm X$ are identically distributed.

Let $X_1$ be the first component of $\bm X$. If $X_1$ has finite variance, denoted by $\sigma^2$, then by using $\var( \bm\th \cdot\bm X ) 
= \var(\bm\eta\cdot \bm X)$, 
we have $\sigma^2 \sum_{i=1}^n \theta_i^2= \sigma^2 \sum_{i=1}^n \eta_i^2$. Since $\bm\th \prec \bm\eta$, this implies $\sigma^2=0$. 

In case $ X_1 $ does not have finite second moment, we can use a strictly convex function $f$ given by $f(x) =\sqrt{x^2+1}$.
Since  $f(x)\le x+1$ for $x\ge 0$,  
 the  two expectations $\E[f(\bm\th \cdot\bm X)] $ and $ \E[f(\bm\eta\cdot \bm X)]$ are finite and equal. 
In view of Lemma \ref{lem:maj}, it suffices to consider the case that $\bm\th $ is a $T$-transform of $\bm\eta$. 
In this case,  
since $\bm \th \ne \bm \eta $, 
there exist $i,j\in [n]$
such that $\eta_i\ne \eta_j$ 
and $\eta_i<\theta_i,\theta_j<\eta_j$
and $\theta_i+\theta_j=\eta_i+\eta_j$. 
Note that $$
\E[f(Z+ \theta_i  X_i
+  \theta_j X_j 
)] = \E[f(Z+  \eta_i X_i +  \eta_j X_j )],
$$
where $Z=\sum_{k\in [n]\setminus\{i,j\}} \theta_k X_k  =\sum_{k\in [n]\setminus\{i,j\}} \eta_k X_k  $.  For all $a,b\in \R$ distinct, due to  strict convexity of $f$, we have 
$$f(Z+ \theta_i  a
+  \theta_j  b 
)+f(Z+ \theta_i  b
+  \theta_j  a 
)<f(Z+  \eta_i a +  \eta_j b )+f(Z+  \eta_i b +  \eta_j a ).$$
Moreover, we have 
$$
\E[f(Z+ \theta_i  X_i
+  \theta_j X_j 
)] = \frac 12 \E[f(Z+ \theta_i  X_i
+  \theta_j X_j 
)]+\frac 12 \E[f(Z+ \theta_i  X_j
+  \theta_j X_i 
)] .
$$
Using  the equality in expectation and independence between $Z,X_i,X_j$, we know that $X_i$ and $X_j$
cannot take distinct values, and they must be degenerate. 
\end{proof}

\subsection{Feller-Pareto random variables}\label{app:FP}

For $\mu\in\R$, $\sigma>0$, and $\gamma>0$, a random variable $W$ is called a \emph{Feller-Pareto} random variable if it has the following stochastic representation,
$$W=\mu+\sigma\(\frac{Z_1}{Z}\)^\gamma,$$
where $Z_1\sim \mathrm{Gamma}(\beta)$ and $Z\sim \mathrm{Gamma}(\alpha)$, $\beta,\alpha>0$, are independent. We write $W\sim {\rm FP} (\alpha, \beta, \mu, \sigma, \gamma)$. Note that $\E(W)=\infty$ if $\alpha\le \gamma$; we refer to \cite{A15} for distributional  properties of the Feller-Pareto family.

The class of Fell-Pareto distributions is very general and it includes Pareto distribution Type I ($FP(\alpha,1,\sigma,\sigma,1)$),   Pareto distribution Type II (${\rm FP}(\alpha,1,\mu,\sigma,1)$), Pareto distribution Type III (${\rm FP}(1,1,\mu,\sigma,\gamma)$), and Pareto distribution Type IV (${\rm FP}(\alpha, 1, \mu,\sigma,\gamma)$). Clearly, Pareto distributions Type I, II, and III are contained in the class of Pareto distributions Type IV with distribution function
$$F(x)=1-\(1+\(\frac{x-\mu}{\sigma}\)^{1/\gamma}\)^{-\alpha},~~~x> \mu.$$
The Feller-Pareto family also includes distributions that are distinctly non-Paretian in their tails; see, e.g., (3.2.17) of \cite{A15}. For our study, as first-order stochastic dominance is location invariant, it suffices to consider the case $\mu=0$ and $\sigma=1$ and we write ${\rm FP}(\alpha, \beta, \gamma) ={\rm FP}(\alpha,\beta,0,1,\gamma)$.

Let $Z_1,\dots,Z_n\sim \mathrm{Gamma}(\beta)$ and $Z\sim \mathrm{Gamma}(\alpha)$, $\beta,\alpha>0$, be independent. Define a random vector via the stochastic representation below
\begin{equation}\label{eq:MFP}
\(X_1,\dots,X_n\)=\(\(\frac{Z_1}{Z}\)^\gamma,\dots,\(\frac{Z_n}{Z}\)^\gamma\).\end{equation} 
Then $\(X_1,\dots,X_n\)$ follows a multivariate Feller-Pareto distribution, denoted by ${\rm MFP}(\alpha,\beta,\gamma)$, with common margins ${\rm FP}(\alpha, \beta, \gamma)$.  Clearly, $X_1,\dots,X_n$ are positively dependent as they are all affected by the common shock $Z^\gamma$. Similar to  ${\rm FP}(\alpha, \beta, \gamma)$, ${\rm MFP}(\alpha,\beta,\gamma)$ includes special cases of  multivariate Pareto distributions, from Type I to IV; the parameters are chosen  analogously to the univariate case. 
\begin{theorem}\label{thm:MFP}
Suppose that $\bm\th, \bm\eta\in \R_+^n$ satisfy $\bm\th \preceq \bm\eta$. Let $\bm X \sim {\rm MFP}(\alpha,\beta,\gamma)$,  $\alpha,\beta,\gamma>0$ and $\alpha\le \gamma$.  Then $\bm{\eta} \cdot\bm X \le_{{\rm st}} \bm{\theta} \cdot \bm X$.
\end{theorem}
\begin{proof}
We write $\bm X = ((Z_1/Z)^\gamma,\dots,(Z_n/Z)^\gamma)$ where  $Z_1,\dots,Z_n\sim \mathrm{Gamma}(\beta)$ and $Z\sim \mathrm{Gamma}(\alpha)$ are independent.    Let  $\bm  Z= (Z_1^\gamma,\dots,Z_n^\gamma)$.
 For $x>0$, we have 
 \begin{align*}
 \p(\bm{\eta} \cdot\bm X\ge x)&=\p\(\frac{\bm{\eta} \cdot \bm Z}{Z^\gamma}\ge x\)=\E\(\p\(Z^\gamma\le \frac{\bm{\eta} \cdot \bm Z}{x}|\bm{\eta} \cdot \bm Z\)\)=\E(\phi(\bm{\eta} \cdot \bm Z)),
 \end{align*}
where $\phi(y)=\p(Z^\gamma\le y/x)$, $y>0$. Next, we show $\phi$ is concave. Assume $x=1$ for simplicity. Taking the second derivative of $\phi$, we get 
$$\phi^{''}(y)=\frac{y^{\alpha/\gamma-2}\exp\(-y^{1/\gamma}\)(\alpha/\gamma-1-y^{1/\gamma}/\gamma)}{\gamma\Gamma(\alpha) }.$$
As $\alpha\le \gamma$, $\phi^{''}<0$ and hence $\phi$ is concave. For two random variables $X$ and $Y$, we say $X$ is smaller than $Y$ in the \emph{convex order}, denoted by $X\le_{\rm cx}Y$, if $\E(u(X))\le \E(u(Y))$ for all convex functions $u$ provided that the expectations exist. By Theorem 3.A.35 of \cite{SS07}, since  $Z_1,\dots,Z_n$ are iid, $\bm{\eta} \cdot \bm Z\ge_{\rm cx} \bm{\theta} \cdot \bm Z.$ By the concavity of $\phi$, 
$$ \p(\bm{\eta} \cdot\bm X\ge x)=\E(\phi(\bm{\eta} \cdot \bm Z))\le \E(\phi(\bm{\theta} \cdot \bm Z))=\p(\bm{\theta} \cdot\bm X\ge x).$$
The proof is done.
\end{proof}

\begin{remark}
If   $\bm X$ in Theorem \ref{thm:MFP}  has a stochastic representation \eqref{eq:MFP} where $Z_1,\dots,Z_n$ are exchangeable, Theorem \ref{thm:MFP} still holds. This is because in the proof of Theorem \ref{thm:MFP}, $\bm{\eta} \cdot \bm Z\ge_{\rm cx} \bm{\theta} \cdot \bm Z$   also holds for exchangeable $Z_1,\dots,Z_n$ (Theorem 3.A.35 of \cite{SS07}). \end{remark}

\section*{Funding}

T.~Hu would like to acknowledge financial support from National Natural Science Foundation of China (No.~72332007, 12371476).
R.~Wang is supported by the Natural Sciences and Engineering Council of Canada (RGPIN-2024-03728; CRC-2022-00141). 
Z.~Zou is supported by National Natural Science Foundation of China (No. 12401625), the China Postdoctoral Science Foundation (No. 2024M753074), the Postdoctoral Fellowship Program of CPSF (GZC20232556), and the Fundamental Research Funds for the Central Universities (No. WK2040000108).

\section*{Disclosure statement}

No potential conflict of interest was reported by the authors.


\begin{thebibliography}{10}
\small

\bibitem[Andriani and McKelvey, 2007]{andriani2007beyond}
Andriani, P. and McKelvey, B. (2007).
\newblock Beyond Gaussian averages: Redirecting international business and management research toward extreme events and power laws.
\newblock {\em Journal of International Business Studies}, \textbf{38}(7):1212--1230.

\bibitem[Arab et al., 2024]{ALO24}
Arab, I., Lando, T., and Oliveira, P. E. (2024). Convex combinations of random variables
stochastically dominate the parent for a large class of heavy-tailed distributions.
\emph{arXiv:2411.14926.}





\bibitem[Arnold, 2015]{A15}
Arnold, B.~C. (2015).
\newblock {\em Pareto Distributions}.
\newblock Second Edition. CRC Press, New York.

\bibitem[\protect\citeauthoryear{Artzner et al.}{Artzner et al.}{1999}]{ADEH99}
{Artzner, P., Delbaen, F., Eber, J.-M. and Heath, D.} (1999). Coherent measures of risk. \emph{Mathematical Finance}, \textbf{9}(3):203--228.
		
\bibitem[Balkema and de Haan, 1974]{BD74}
Balkema, A. and de Haan, L. (1974).
\newblock Residual life time at great age.
\newblock {\em Annals of Probability}, \textbf{2}(5):792--804.

\bibitem[\protect\citeauthoryear{Chen et al.}{2024a}]{CEW24}
Chen, Y., Embrechts, P. and Wang, R. (2024a). An unexpected stochastic dominance: Pareto distributions, dependence, and diversification. \emph{Operations Research}, forthcoming.

\bibitem[\protect\citeauthoryear{Chen et al.}{2024b}]{CEW24(2)}
Chen, Y., Embrechts, P. and Wang, R. (2024b). Risk exchange under infinite-mean Pareto models. \emph{arXiv:2403.20171}.

\bibitem[Chen and Shneer, 2024]{CS24}
Chen, Y. and Shneer, S. (2024). Risk aggregation and stochastic dominance for a class of heavy-tailed distributions. arXiv preprint \emph{arXiv:2408.15033.}

\bibitem[Choi and Lee, 2020]{choi2020power}
Choi, M. and Lee, C.-Y. (2020).
\newblock Power-law distributions of corporate innovative output: evidence from
  U.S. patent data.
\newblock {\em Scientometrics}, \textbf{122}(1):519--554.

\bibitem[\protect\citeauthoryear{Embrechts et al.}{Embrechts et al.}{1997}]{EKM97}
Embrechts, P., Kl\"uppelberg, C. and Mikosch, T. (1997). \emph{Modelling Extremal Events for Insurance and Finance}. Springer, Heidelberg.

\bibitem[Embrechts et~al., 2002]{embrechts2002correlation}
Embrechts, P., McNeil, A. and Straumann, D. (2002).
\newblock Correlation and dependence in risk management: properties and pitfalls.
In \newblock {\em Risk Management: Value at Risk and Beyond} (Eds: Dempster), pp.~176--223, Cambridge University Press.

\bibitem[Fama and Miller, 1972]{FM72}
Fama, E.~F. and Miller, M.~H. (1972).
\newblock {\em The Theory of Finance}.
\newblock Dryden Press, Hinsdale.

\bibitem[\protect\citeauthoryear{F\"ollmer and Schied}{F\"ollmer and Schied}{2016}]{FS16} F\"ollmer, H.~and Schied, A.~(2016). \emph{Stochastic Finance: An Introduction in Discrete Time}. Fourth Edition.  {Walter de Gruyter, Berlin}.

\bibitem[Hofert and W{\"u}thrich, 2012]{hofert2012statistical}
Hofert, M. and W{\"u}thrich, M.~V. (2012).
\newblock Statistical review of nuclear power accidents.
\newblock {\em Asia-Pacific Journal of Risk and Insurance}, \textbf{7}(1), Article 1.

\bibitem[Ibragimov, 2005]{ibragimov2005new}
Ibragimov, R. (2005).
\newblock { New majorization theory in economics and martingale convergence results in econometrics}.
\newblock Ph.D. Dissertation, Yale University, New Haven, CT.

\bibitem[Ibragimov, 2009]{ibragimov2009portfolio}
Ibragimov, R. (2009).
\newblock Portfolio diversification and value at risk under thick-tailedness.
\newblock {\em Quantitative Finance}, \textbf{9}(5):565--580.

\bibitem[Ibragimov et~al., 2015]{ibragimov2015heavy}
Ibragimov, M., Ibragimov, R. and Walden, J. (2015).
\newblock {\em Heavy-Tailed Distributions and Robustness in Economics and
Finance}. Springer Cham.

\bibitem[\protect\citeauthoryear{Ibragimov et al.}{2009}]{IJW09}
Ibragimov, R., Jaffee, D. and Walden, J. (2009). Non-diversification traps in markets for catastrophic risk. \emph{Review of Financial Studies}, \textbf{22}(3):959--993.

\bibitem[\protect\citeauthoryear{Ma}{1998}]{Ma98}
Ma, C. (1998). On peakedness of distributions of convex combinations. \emph{Journal of Statistical Planning and Inference}, \textbf{70}(1):51-56.

\bibitem[\protect\citeauthoryear{Markowitz}{Markowitz}{1952}]{M52}
Markowitz, H. (1952). Portfolio selection. \emph{Journal of Finance}, \textbf{7}(1):77--91.

\bibitem[\protect\citeauthoryear{Marshall et~al.}{Marshall et~al.}{2011}]{MOA11}
Marshall, A.~W., Olkin, I. and Arnold, B. (2011).
{\em Inequalities: Theory of Majorization and Its Applications}, 2nd edition. Springer, New York.

\bibitem[McNeil et~al., 2015]{MFE15}
McNeil, A.~J., Frey, R. and Embrechts, P. (2015).
\newblock {\em Quantitative Risk Management: Concepts, Techniques and Tools}. Revised edition.
\newblock Princeton University Press.

\bibitem[Moscadelli, 2004]{moscadelli2004modelling}
Moscadelli, M. (2004).
\newblock The modelling of operational risk: Experience with the analysis of the data collected by the Basel committee. Technical Report 517. \emph{SSRN}:557214.


\bibitem[M\"uller, 2024]{M24}
M\"uller, A. (2024). Some remarks on the effect of risk sharing and diversification for
infinite mean risks. \emph{arXiv:2411.10139}


\bibitem[\protect\citeauthoryear{Nelsen}{Nelsen}{2006}]{N06} {Nelsen, R.} (2006). \emph{An Introduction to Copulas}.  Second Edition. Springer, New York.

\bibitem[Pickands, 1975]{P75}
Pickands, J. (1975).
\newblock
Statistical inference using extreme order statistics.
\newblock {\em Annals of Statistics}, \textbf{3}(1):119--131.


\bibitem[\protect\citeauthoryear{Puccetti and Wang}{Puccetti and
  Wang}{2015}]{PW15}
Puccetti, G. and Wang R.  (2015).
Extremal dependence concepts.
 \emph{Statistical Science},  \textbf{30}(4):485--517.


 

\bibitem[\protect\citeauthoryear{Rothschild and Stiglitz}{1970}]{RS70}
Rothschild, M. and Stiglitz, J. E. (1970). Increasing risk: I. A definition. \emph{Journal of Economic Theory}, \textbf{2}(3):225--243.



\bibitem[Samuelson, 1967]{S67}
Samuelson, P. A. (1967). General proof that diversification pays. \emph{Journal of Financial and Quantitative Analysis}, \textbf{2}(1):1--13.

\bibitem[Sarabia et~al., 2016]{SGPJ16}
Sarabia, J.~M., G{\'o}mez-D{\'e}niz, E., Prieto, F. and Jord{\'a}, V. (2016).
\newblock Risk aggregation in multivariate dependent Pareto distributions.
\newblock {\em Insurance: Mathematics and Economics}, \textbf{71}:154--163.

\bibitem[Scherer et~al., 2001]{scherer2001uncertainty}
Scherer, F.~M., Harhoff, D. and Kukies, J. (2001).
\newblock Uncertainty and the size distribution of rewards from innovation.
\newblock In {\em Capitalism and Democracy in the 21st Century: Proceedings of
  the International Joseph A. Schumpeter Society Conference, Vienna 1998
  “Capitalism and Socialism in the 21st Century”}, pages 181--206.
  Springer.

\bibitem[\protect\citeauthoryear{Shaked and Shanthikumar}{Shaked and Shanthikumar}{2007}]{SS07}
Shaked, M. and Shanthikumar, J.~G. (2007).  {\em Stochastic Orders}. Springer, New York.



\bibitem[Silverberg and Verspagen, 2007]{silverberg2007size}
Silverberg, G. and Verspagen, B. (2007).
\newblock The size distribution of innovations revisited: An application of extreme value statistics to citation and value measures of patent significance.
\newblock {\em Journal of Econometrics}, \textbf{139}(2):318--339.

\end{thebibliography}
\end{document}